\documentclass[12pt]{article}
\usepackage{amsmath}
\usepackage{amssymb}
\usepackage{amsthm}
\usepackage{graphicx}
\usepackage{subfigure}
\usepackage{changepage}
\usepackage{hyperref}
\usepackage{cite}
\usepackage{url}
\usepackage{breakurl}
\usepackage[small]{caption}
\usepackage{mathtools,xparse}
\theoremstyle{definition}
\newtheorem{theorem}{Theorem}
\newtheorem{corollary}{Corollary}

\begin{document}
\baselineskip 0.6cm

\def\bra#1{\langle #1 |}
\def\ket#1{| #1 \rangle}
\def\inner#1#2{\langle #1 | #2 \rangle}
\def\app#1#2{%
  \mathrel{%
    \setbox0=\hbox{$#1\sim$}%
    \setbox2=\hbox{%
      \rlap{\hbox{$#1\propto$}}%
      \lower1.1\ht0\box0%
    }%
    \raise0.25\ht2\box2%
  }%
}
\def\approxprop{\mathpalette\app\relax}
\DeclarePairedDelimiter{\norm}{\lVert}{\rVert}

\begin{titlepage}

\begin{flushright}
\end{flushright}

\vskip 1.2cm

\begin{center}
{\Large \bf Spacetime from Unentanglement}

\vskip 0.7cm

{\large Yasunori Nomura,$^{a,b,c}$
  Pratik Rath,$^{a,b}$ and
  Nico Salzetta$^{a,b}$}

\vskip 0.5cm

$^a$ {\it Berkeley Center for Theoretical Physics, Department of Physics,\\
  University of California, Berkeley, CA 94720, USA}

\vskip 0.2cm

$^b$ {\it Theoretical Physics Group, Lawrence Berkeley National Laboratory, 
 CA 94720, USA}

\vskip 0.2cm

$^c$ {\it Kavli Institute for the Physics and Mathematics of the Universe 
 (WPI), University of Tokyo, Kashiwa, Chiba 277-8583, Japan}

\vskip 0.8cm

\abstract{The past decade has seen a tremendous effort toward unraveling 
the relationship between entanglement and emergent spacetime.  These 
investigations have revealed that entanglement between holographic 
degrees of freedom is crucial for the existence of bulk spacetime. 
We examine this connection from the other end of the entanglement 
spectrum and clarify the assertion that maximally entangled states 
have no reconstructable spacetime.  To do so, we first define the 
conditions for bulk reconstructability.  Under these terms, we scrutinize 
two cases of maximally entangled holographic states.  One is the 
familiar example of AdS black holes; these are dual to thermal states 
of the boundary CFT.  Sending the temperature to the cutoff scale makes 
the state maximally entangled and the respective black hole consumes 
the spacetime.  We then examine the de~Sitter limit of FRW spacetimes. 
This limit is maximally entangled if one formulates the boundary theory 
on the holographic screen.  Paralleling the AdS black hole, we find 
the resulting reconstructable region of spacetime vanishes.  Motivated 
by these results, we prove a theorem showing that maximally entangled 
states have no reconstructable spacetime.  Evidently, the emergence 
of spacetime is endemic to intermediate entanglement.  By studying 
the manner in which intermediate entanglement is achieved, we uncover 
important properties about the boundary theory of FRW spacetimes. 
With this clarified understanding, our final discussion elucidates 
the natural way in which holographic Hilbert spaces may house states 
dual to different geometries.  This paper provides a coherent picture 
clarifying the link between spacetime and entanglement and develops 
many promising avenues of further work.}

\end{center}
\end{titlepage}

\tableofcontents
\newpage

\section{Introduction}
\label{sec:intro}

It is believed that dynamical spacetime described by general relativity 
is an emergent phenomenon in the fundamental theory of quantum gravity. 
Despite this pervasive idea, the materialization of spacetime itself 
is not fully understood.  Holography posits that a fundamental 
description of quantum gravity resides in a non-gravitational 
spacetime whose dimension is less than that of the corresponding 
bulk spacetime~\cite{'tHooft:1993gx,Susskind:1994vu,Bousso:2002ju}. 
In this paper, we study the emergence of gravitational spacetime 
in the context of holography, using the renowned anti-de~Sitter 
(AdS)/conformal field theory (CFT) correspondence~\cite{Maldacena:1997re} 
and a putative holographic theory of Friedmann-Robertson-Walker (FRW) 
spacetimes~\cite{Nomura:2016ikr}.

In this paper, we expound on the intimate relationship between the 
emergence of spacetime and the lack of maximal entanglement in the 
boundary state.  Through this, we see that the existence of spacetime 
is necessarily non-generic and that nature seizes the opportunity to 
construct local spacetime when states deviate from maximal entanglement. 
A reason why this viewpoint is not heavily emphasized (see, however, 
e.g.\ Refs.~\cite{Nomura:2017npr,Bao:2017gza}) in the standard context 
of AdS/CFT is that one almost always considers states with energy much 
lower than the cutoff (often sent to infinity). The restriction to 
these ``low energy'' states implicitly narrows our perspective to those 
automatically having non-maximal entropy.  However, in a holographic 
theory with a finite cutoff scale (or a fundamentally nonlocal theory), 
the regime of maximal entropy is much more readily accessible.  This 
happens to be the case in FRW holography, and perhaps holography in 
general.  Through this lens, we analyze the emergence of spacetime 
both in the familiar setting of Schwarzschild-AdS spacetime with 
an infrared cutoff and in flat FRW universes.  We explicitly see that 
the directly reconstructable region of spacetime~\cite{Nomura:2017npr,%
Sanches:2017xhn} emerges only as we deviate from maximally entangled 
states.  This implies that a holographic theory of exact de~Sitter 
space cannot be obtained as a natural limit of theories dual to FRW 
spacetimes by sending the fluid equation of state parameter, $w$, to 
$-1$.  In addition to analyzing these two examples, we prove a theorem 
demonstrating the lack of directly reconstructable spacetime in the 
case that a boundary state is maximally entangled.

After surveying the relationship between spacetime and (the lack of) 
entanglement, we then analyze the deviation from maximal entropy itself. 
The size of the subregions for which deviations occur reveals valuable 
information about the underlying holographic theory, and observing the 
corresponding emergence of spacetime in the bulk provides a glimpse into 
the mechanism by which nature creates bulk local degrees of freedom. 
In the case of Schwarzschild-AdS, reconstructable spacetime (the region 
between the horizon and the cutoff) appears as the temperature in the 
local boundary theory (the CFT) is lowered, and the resulting entanglement 
entropy structure (calculated holographically) is consistent with a 
local theory at high temperature.  However, this entanglement structure 
is not observed in the case of FRW spacetimes as we adjust $w$ away 
from $-1$; additionally, the reconstructable region grows from the 
deepest points in the bulk outward.  This suggests that the manner 
in which entanglement is scaffolded is unlike that of AdS/CFT.  In 
fact, this aberrant behavior leads us to believe that the holographic 
theory dual to FRW spacetimes has nonlocal interactions.

The relationship between spacetime and 
quantum entanglement between holographic degrees of freedom is no 
secret~\cite{Ryu:2006bv,Hubeny:2007xt,Swingle:2009bg,VanRaamsdonk:2009ar,%
Faulkner:2013ana,Engelhardt:2014gca,Sanches:2016sxy}, but what {\it is} 
spacetime?  Undoubtedly, entanglement is a necessity for the existence 
of spacetime.  But, it is indeed possible to have too much of a good 
thing.  The analysis here exposes the inability to construct spacetime 
from maximally entangled boundary states.  Since typical states in 
a Hilbert space are maximally entangled~\cite{Page:1993df}, this implies 
that states with bulk dual are not typical.  We see that spacetime is 
an emergent property of non-generic states in the Hilbert space with 
both non-vanishing and non-maximal entanglement for subregions.  The 
existence of entanglement allows for the construction of a code subspace 
of states~\cite{Almheiri:2014lwa} in which local, semi-classical bulk 
degrees of freedom can be encoded redundantly.  Simultaneously, the 
lack of maximal entanglement allows for a code subspace with subsystem 
recovery---hence partitioning the bulk into a collection of local Hilbert 
spaces.  With this perspective, we see that holographic theories are 
exceedingly enterprising---once deviating from maximal entanglement, 
nature immediately seizes the opportunity to construct spacetime.  In 
this sense, spacetime is the byproduct of nature's efficient use of 
intermediate entanglement to construct codes with subsystem recovery.

For a given spacetime with a holographic boundary, one can calculate the 
von~Neumann entropies for all possible subregions of the boundary via 
the Hubeny-Rangamani-Ryu-Takayanagi (HRRT) prescription~\cite{Ryu:2006bv,%
Hubeny:2007xt,Sanches:2016sxy}.  The corresponding entanglement structure 
heavily constrains the possible boundary states, but by no means uniquely 
specifies it.  In fact, given an entanglement structure and a tensor 
product Hilbert space, one can always find a basis for the Hilbert 
space in which all basis states have the desired entanglement structure. 
If one considers each of these basis states to be dual to the spacetime 
reproducing the entanglement, then by superpositions one could entirely 
change the entanglement structure, and hence the spacetime.  This 
property naturally raises the question of how the boundary Hilbert 
space can accommodate states dual to different semiclassical geometries. 
Fortunately, for generic dynamical systems, the Hilbert space can be 
binned into energy bands, and canonical typicality provides us with 
the result that generic states {\it within these bands} have the same 
entanglement structure, regardless of the energy band's size.  This 
allows the holographic Hilbert space to contain states dual to many 
different spacetimes, each of which can have bulk excitations encoded 
state independently.  Importantly, this is contingent on the result 
that typical states have no spacetime.

\subsubsection*{Outline}

Section~\ref{sec:typical} walks through the statement that maximally 
entangled (and hence typical) states have no reconstructable spacetime. 
This is broken down into parts.  First, we must define what we mean by 
reconstructable; this is detailed in Section~\ref{subsec:reconst}, and 
is very important toward understanding the framework of the rest of 
the paper. We then use this construction in Section~\ref{subsec:AdS-BH} 
to investigate the reconstructable region of AdS with a black hole. 
We see the expected behavior that the reconstructable region of spacetime 
vanishes as the temperature of the black hole reaches the cutoff scale, 
making the state typical.  In Section~\ref{subsec:dS_max-ent}, we 
show that de~Sitter states are maximally entangled by finding their 
HRRT surfaces.  In Section~\ref{subsec:dS_no-spacetime}, we combine 
numerical results for flat FRW universes and use the additional property 
that de~Sitter's HRRT surfaces lie on a null cone to show that the 
reconstructable region vanishes in the de~Sitter limit of FRW spacetimes. 
Motivated by these results, in Section~\ref{subsec:proof} we prove 
a theorem showing that if a state is maximally entangled, then its HRRT 
surfaces either wrap the holographic space or live on the null cone. 
This is then used to present the general argument that maximally 
entangled states have no spacetime.

Section~\ref{sec:atypical} compares the emergence of spacetime in 
the two theories we are considering. Sections~\ref{subsec:lower-T} 
and \ref{subsec:FRW-Q} present results comparing the entanglement 
structure of AdS black holes and FRW spacetimes, respectively. 
Section~\ref{subsec:locality} interprets these results and argues 
that the appropriate holographic dual of FRW spacetimes is most 
likely nonlocal.

In Section~\ref{sec:linearity}, we put together all of the previous 
results and explain how one Hilbert space can contain states dual to 
many different semiclassical spacetimes.  Here we discuss the lack of 
a need for state dependence when describing the directly reconstructable 
region.

In Appendix~\ref{app:two-sided}, we analyze two-sided black holes 
within our construction and discuss how a version of complementarity 
works in this setup.  Appendices~\ref{app:S-AdS} and \ref{app:dS-FRW} 
collect explicit calculations for Schwarzschild-AdS and the de Sitter 
limit of FRW spacetimes, respectively.

\section{Maximally Entropic States Have No Spacetime}
\label{sec:typical}

In this section, we see that maximally entangled states in holographic 
theories do not have directly reconstructable spacetime.  First we 
lay out the conditions for reconstructability in general theories of 
holographic spacetimes.  Then we examine the familiar example of a 
large static black hole in AdS and determine its reconstructable region. 
We then discuss the de~Sitter limit of flat FRW spacetimes.  Finally, 
we prove a theorem establishing that maximally entropic holographic 
states have no reconstructable spacetime.

\subsection{Holographic reconstructability}
\label{subsec:reconst}

In order to argue that typical states have no reconstructable region, 
we must first present the conditions for a region of spacetime to be 
reconstructed from the boundary theory.  We adopt the formalism presented 
first in Ref.~\cite{Sanches:2017xhn} but appropriately generalized 
in Ref.~\cite{Nomura:2017npr} to theories living on holographic 
screens~\cite{Bousso:1999cb} (which naturally includes the boundary 
of AdS as in the AdS/CFT correspondence).

The question to answer is:\ ``given a boundary state and its time 
evolution with a known gravitational bulk dual, what regions of 
the bulk can be reconstructed?''  This may sound tautological, but 
it is not.  Settings in which this question is nontrivial include 
spacetimes with black holes and other singularities.  From entanglement 
wedge reconstruction~\cite{Jafferis:2015del,Dong:2016eik}, we know 
that the information of a pure black hole is contained in the boundary 
theory but whether or not the interior is reconstructable is unknown. 
In holographic theories of general spacetimes, we are interested in 
describing spacetimes with big bang singularities and a natural question 
is whether or not the theory reconstructs spacetime arbitrarily close 
to the initial singularity.

To answer this question, Ref.~\cite{Sanches:2017xhn} proposed that 
reconstructable points in a spacetime are precisely those that can be 
localized at the intersection of entanglement wedges.  This is similar 
to the proposal in Ref.~\cite{Kabat:2017mun} which advocates that 
reconstructable points are those located at the intersection of 
HRRT surfaces anchored to arbitrary achronal subregions of the AdS 
conformal boundary.  However, this construction lacks the ability 
to localize points in entanglement shadows, which can form in rather 
tame spacetimes (e.g.\ a neutron star in AdS), while using the 
intersection of entanglement wedges allows us to probe these regions.

In order to generalize this to theories living on holographic screens, 
an essential change is that one can only consider HRRT surfaces anchored 
to the leaves of a given holographic screen (usually associated to 
a fixed reference frame)~\cite{Nomura:2017npr}.  This is because 
holographic screens have a unique foliation into leaves that corresponds 
to a particular time foliation of the holographic theory.  Thus the 
von~Neumann entropy of subregions in the holographic theory only makes 
sense for subregions of a single leaf.  Note that despite the lack of 
a unique time foliation of the conformal boundary, this subtlety is 
also present in AdS/CFT.  Namely, one should consider only a single 
time foliation of the boundary and the HRRT surfaces anchored to 
the associated equal time slices even in asymptotically AdS 
spacetimes~\cite{Nomura:2017npr}.%
\footnote{This is related to the work in Ref.~\cite{Kusuki:2017jxh}, 
 which studied the breakdown of the HRRT formula in certain limits 
 of boundary subregions.  These breakdowns correspond to disallowed 
 foliations of the boundary theory.}
This issue becomes manifest when the boundary contains multiple 
disconnected components, as we discuss in Appendix~\ref{app:two-sided}.

Thus we define the reconstructable region of a spacetime as the 
union of all points that can be localized at the boundary of 
entanglement wedges of all subregions of leaves of the holographic 
screen.  Henceforth, we will refer to the regions of spacetime constructed 
in this way as the directly reconstructable regions (or simply the 
reconstructable regions when the context is clear), and our analysis 
will primarily focus on these regions.  For a more detailed study 
of directly reconstructable regions in general spacetimes, see 
Ref.~\cite{Nomura:2017npr}.  In particular, this definition only 
allows for the reconstruction of points outside the horizon for 
a quasi-static one-sided black hole, since such a horizon acts as 
an extremal surface barrier~\cite{Engelhardt:2013tra}.%
\footnote{This does not exclude the possibility that the holographic 
 theory allows for some effective description of regions other than 
 the directly reconstructable one, e.g. the black hole interior 
 (perhaps along the lines of Ref.~\cite{Papadodimas:2015jra}).  This 
 may make the interior spacetime manifest, perhaps at the cost of 
 losing the local description elsewhere, and may be necessary to 
 describe the fate of a physical object falling into a black hole. 
 We focus on spacetime regions that can be described by the boundary 
 theory without resorting to such descriptions.}
This also prevents the direct reconstruction of points near singularities 
such as big bang singularities and the black hole singularity of a 
two-sided black hole.

Now that we have detailed the conditions for regions of spacetime 
to be directly reconstructable, we must determine a measure of ``how 
much'' spacetime is reconstructable.  This will allow us to see the 
loss of spacetime in the limit of states becoming typical.  In the 
context of quantum error correction~\cite{Almheiri:2014lwa}, we are 
attempting to quantify the factorization of the code subspace, e.g.\ 
how many dangling bulk legs exist in a tensor network representation 
of the code~\cite{Pastawski:2015qua,Hayden:2016cfa}.  We expect 
the spacetime volume of the reconstructable region to be indicative 
to this property, and we will use it in our subsequent analyses. 
The bulk spacetime directly reconstructable from a single leaf 
depends on features of the bulk, for example, the existence of 
shadows and time dependence.  In the case of $(d+1)$-dimensional 
flat FRW spacetimes, we find that a codimension-0 region can be 
reconstructed from a single leaf.  On the other hand, in any static 
spacetime, all HRRT surfaces anchored to one leaf live in the same 
time slice in the bulk, and hence their intersections reconstruct 
a codimension-1 surface of the bulk.  This is the case in an AdS 
black hole.

The discrepancy of the dimensions of the directly reconstructable 
regions for different spacetimes of interest may seem to cause issues 
when trying to compare the loss of spacetime in these systems.  Namely, 
it seems difficult to compare the loss of reconstructable spacetime 
in Schwarzschild-AdS as we increase the black hole mass to the loss 
of spacetime in the $w \rightarrow -1$ limit of flat FRW spacetimes. 
However, in all cases, the spacetime region directly reconstructable 
from a small time interval in the boundary theory is codimension-0. 
We can then examine the relative loss of spacetime in both cases (black 
hole horizon approaching the boundary in AdS space and $w \rightarrow -1$ 
in FRW spacetimes) by taking the ratio of the volume of the reconstructable 
region to the reconstructable volume of some reference state (e.g.\ 
pure AdS and flat FRW with some fixed $w \neq -1$).  In static 
spacetimes, this will reduce to a ratio of the spatial volumes 
reconstructed on a codimension-1 slice, allowing us to consider 
only the volume of regions reconstructed from single leaves.

\subsection{Large AdS black holes}
\label{subsec:AdS-BH}

Here we will see how spacetime disappears as we increase the mass 
of the black hole in static Schwarzschild-AdS spacetime, making the 
corresponding holographic state maximally entangled.  We consider 
a holographic pure state living on the (single) conformal boundary 
of AdS.  We introduce an infrared cutoff $r \leq R$ in AdS space and 
consider a $d+1$ dimensional large black hole with horizon radius 
$r = r_+$.

As discussed in Section~\ref{subsec:reconst}, the size of the spacetime 
region directly reconstructable from the boundary theory is characterized 
by $V(r_+,R)$, the spatial volume between the black hole horizon and 
the cutoff.  We normalize it by the volume of the region $r \leq R$ in 
empty AdS space, $V(R)$, to get the ratio
\begin{equation}
  f\Bigl(\frac{r_+}{R}\Bigr) \equiv \frac{V(r_+,R)}{V(R)} 
  = (d-1) \frac{r_+^{d-1}}{R^{d-1}}\! 
    \int_1^{\frac{R}{r_+}}\! \frac{x^{d-2}}{\sqrt{1-\frac{1}{x^d}}}\, dx,
\label{eq:vol-ratio}
\end{equation}
which depends only on $r_+/R$ (and $d$).  As expected, it behaves as
\begin{equation}
  f\Bigl(\frac{r_+}{R}\Bigr) \,\left\{\! \begin{array}{ll}
      \simeq 1      & (r_+ \ll R) \\
      \rightarrow 0 & (r_+ \rightarrow R),
    \end{array} \right.
\label{eq:f-limits}
\end{equation}
in the two opposite limits.  The details of this calculation can be 
found in Appendix~\ref{subapp:S-AdS_volume}.  Here, we plot $f(r_+/R)$ 
in Fig.~\ref{fig:f} for various values of $d$.
\begin{figure}[t]
\begin{center}
  \includegraphics[height=6.5cm]{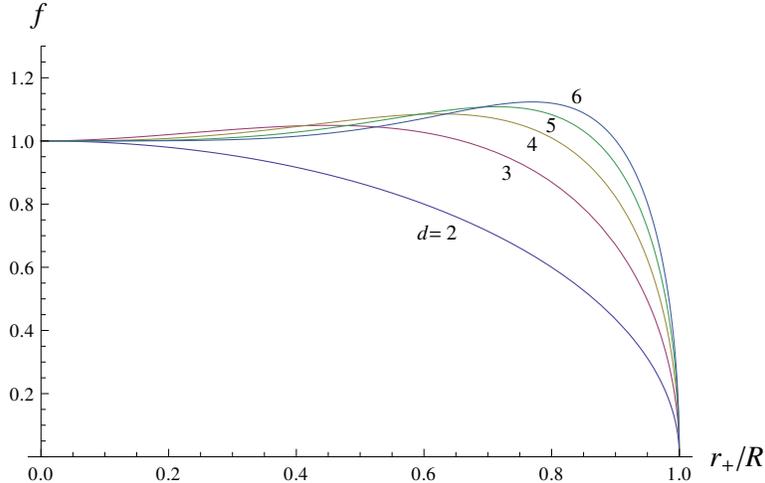}
\end{center}
\caption{The volume $V(r_+,R)$ of the Schwarzschild-AdS spacetime 
 that can be reconstructed from the boundary theory, normalized by the 
 corresponding volume $V(R)$ in empty AdS space:\ $f = V(r_+,R)/V(R)$. 
 Here, $R$ is the infrared cutoff of $(d+1)$-dimensional AdS space, 
 and $r_+$ is the horizon radius of the black hole.}
\label{fig:f}
\end{figure}

In the limit $r_+ \rightarrow R$, the HRRT surface, $\gamma_A$, anchored 
to the boundary of subregion $A$ of a boundary space (a constant $t$ 
slice of the $r = R$ hypersurface) becomes the region $A$ itself or the 
complement, $\bar{A}$, of $A$ on the boundary space, whichever has the 
smaller volume.%
\footnote{We do not impose a homology constraint, since we consider 
 a pure state in the holographic theory.  Additionally, we only consider 
 subregions larger than the cutoff size.}
This implies that the entanglement entropy of $A$, given by the area 
of the HRRT surface as $S_A = \norm{\gamma_A}/4 l_{\rm P}^{d-1}$, becomes 
exactly proportional to the smaller of the volumes of $A$ and $\bar{A}$ 
in the boundary theory:
\begin{equation}
  S_A = \frac{1}{4 l_{\rm P}^{d-1}} 
    {\rm min}\{ \norm{A}, \norm{\bar{A}} \}.
\label{eq:max-ent}
\end{equation}
Here, $\norm{x}$ represents the volume of the object $x$ (often called 
the area for a codimension-2 surface in spacetime), and $l_{\rm P}$ 
is the $(d+1)$-dimensional Planck length in the bulk.  Via usual 
thermodynamic arguments, we interpret this to mean that the state 
in the boundary theory is generic, so that it obeys the Page 
law~\cite{Page:1993df}.%
\footnote{Page's analysis tells us that for a generic state (a Haar 
 random state) in a Hilbert space, the entanglement entropy of a reduced 
 state is nearly maximal.  In fact, at the level of the approximation 
 we employ in this paper, $\norm{A}/l_{\rm P}^{d-1} \rightarrow \infty$, 
 such a state has the maximal entanglement entropy for any subregion, 
 Eq.~(\ref{eq:max-ent}).}
This in turn implies that the temperature of the system, which is 
identified as the Hawking temperature $T_{\rm H}$, is at the cutoff 
scale.%
\footnote{When we refer to a high temperature state, we do not mean 
 that the whole holographic state is a mixed thermal state.  What we 
 really mean is a high energy state, since we focus on pure states.}
$T_{\rm H}$ is related to $r_+$ by
\begin{equation}
  \frac{r_+}{R} = \frac{4\pi l^2}{d R} T_{\rm H},
\label{eq:rR-temp}
\end{equation}
where $l$ is the AdS radius.  Hence, the cutoff scale of the boundary 
theory is given by~\cite{Susskind:1998dq}
\begin{equation}
  \Lambda = \frac{d R}{4\pi l^2}.
\label{eq:cutoff}
\end{equation}
This allows us to interpret the horizontal axis of Fig.~\ref{fig:f} as 
$T_{\rm H}/\Lambda$ from the viewpoint of the boundary theory.

We finally make a few comments.  First, it is important to note that by 
the infrared cutoff, we do not mean that the spacetime literally ends 
there as in the scenario of Ref.~\cite{Randall:1999vf}.  Such termination 
of spacetime would introduce dynamical gravity in the holographic theory, 
making the maximum entropy of a subregion scale as the area, rather than 
the volume, in the holographic theory.  Rather, our infrared cutoff here 
means that we focus only on the degrees of freedom in the bulk deeper 
than $r = R$, corresponding to setting the sliding renormalization 
scale to be $\approx R/l^2$ in the boundary theory.  In particular, 
the boundary theory is still non-gravitational.

Second, to state that spacetime disappears in the limit where the 
holographic state becomes typical, it is crucial to define spacetime 
as the directly reconstructable region.  This becomes clear by 
considering a large subregion $A$ on the boundary theory such that 
$A$ and its HRRT surface $\gamma_A$ enclose the black hole at the 
center.  If we take the simple viewpoint of entanglement wedge 
reconstruction, this would say that spacetime does not disappear 
even if the black hole becomes large and its horizon approaches 
the cutoff surface, since the black hole interior is within the 
entanglement wedge of $A$ so that it still exists in the sense of 
entanglement wedge reconstruction.  We, however, claim that such 
a region does not exist as a localizable spacetime region, as 
explained in Section~\ref{subsec:reconst}.

Third, the curves in Fig.~\ref{fig:f} are not monotonically decreasing 
as $r_+$ increases for $d > 2$, despite the fact that
\begin{equation}
  \frac{d}{dr_+} \bigl\{ S_{A,{\rm max}} - S_{A,{\rm BH}}(r_+) \bigr\} 
  < 0.
\label{eq:resid-S}
\end{equation}
Here, $S_{A,{\rm max}}$ and $S_{A,{\rm BH}}(r_+)$ are the maximal entropy 
and the entropy corresponding to the black hole geometry of subregion 
$A$, given by
\begin{equation}
  S_{A,{\rm max}} = \frac{\norm{A}}{4 l_{\rm P}^{d-1}},
\qquad
  S_{A,{\rm BH}}(r_+) 
  = \frac{\norm{A}}{4 l_{\rm P}^{d-1}} \frac{r_+^{d-1}}{R^{d-1}}.
\label{eq:S_max-BH}
\end{equation}
This increase in spacetime volume may be demonstrating that the 
additional entanglement in the boundary state allows for more bulk 
nodes in the code subspace.  Alternatively, this may be a feature 
of using volume as our measure.  Regardless, the decrease observed 
near the cutoff temperature is the main focus of our attention, and 
we expect any other reasonable measure to correspondingly vanish.

Finally, the statement that spacetime disappears as the holographic state 
approaches typicality persists for two-sided black holes.  In this setup, 
there is a new issue that does not exist in the case of single-sided 
black holes:\ the choice of a reference frame associated with a relative 
time shift between the two boundaries.  The discussion of two-sided 
black holes is given in Appendix~\ref{app:two-sided}.

\subsection{de~Sitter states are maximally entropic}
\label{subsec:dS_max-ent}

We have seen that a large black hole in AdS with $r_+ \rightarrow R$ 
corresponds to CFT states at the cutoff temperature, and that the 
holographic states in this limit have the entanglement entropy structure 
of Eq.~(\ref{eq:max-ent}).  Below, we refer to states exhibiting 
Eq.~(\ref{eq:max-ent}) as the {\it maximally entropic states}.  Is 
there an analogous situation in the holographic theory of FRW spacetimes, 
described in Ref.~\cite{Nomura:2016ikr}?  Here we argue that the de~Sitter 
limit ($w \rightarrow -1$) in flat FRW universes provides one.%
\footnote{For a simple proof applicable to $2+1$ dimensions, see 
 Appendix~\ref{subapp:2Dproof}.}

We first see that the holographic state becomes maximally entropic 
in the case that a universe approaches de~Sitter space at late 
times~\cite{Sanches:2016sxy}.  This situation arises when the universe 
contains multiple fluid components including one with $w = -1$, so 
that it is dominated by the $w = -1$ component at late times.  This 
analysis does not apply directly to the case of a single component with 
$w = -1 + \epsilon$ ($\epsilon \rightarrow 0^+$), which will be discussed 
later.

In the universe under consideration, the FRW metric approaches the 
de~Sitter metric in flat slicing at late times
\begin{equation}
  ds^2 = -dt^2 + e^{\frac{2t}{\alpha}} 
    \bigl( dr^2 + r^2 d\Omega_{d-1}^2 \bigr),
\label{eq:dS-flat}
\end{equation}
where $\alpha$ is the Hubble radius, and we have taken the spacetime 
dimension of the bulk to be $d+1$.  The Penrose diagram of this 
spacetime is depicted in Fig.~\ref{fig:dS-flat}, where constant time 
slices are drawn and the region covered by the coordinates is shaded; 
future timelike infinity $I_+$ corresponds to $t = \infty$, while the 
null hypersurface $N$ corresponds to $t = -\infty$.
\begin{figure}[t]
\begin{center}
  \includegraphics[height=6.5cm]{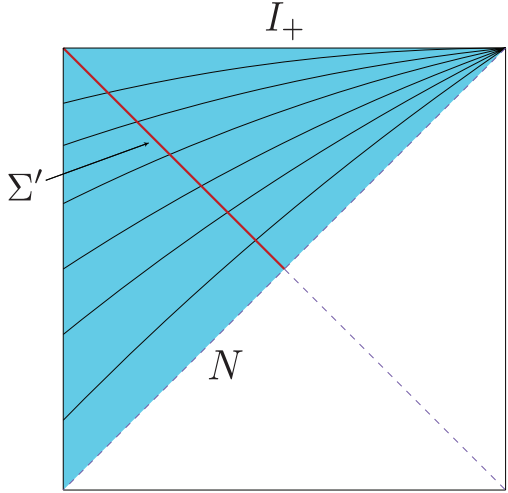}
\end{center}
\caption{The Penrose diagram of de~Sitter space.  The spacetime region 
 covered by the flat-slicing coordinates is shaded, and constant time 
 slices in this coordinate system are drawn.  The codimension-1 null 
 hypersurface $\Sigma'$ is the cosmological horizon for an observer 
 at $r=0$, to which the holographic screen of the FRW universe 
 asymptotes in the future.}
\label{fig:dS-flat}
\end{figure}
At late times, the past holographic screen of the FRW universe asymptotes 
to the codimension-1 null hypersurface $\Sigma'$ depicted in the figure. 
This hypersurface is located at
\begin{equation}
  r = \alpha\, e^{-\frac{t}{\alpha}},
\label{eq:r_dS}
\end{equation}
which corresponds to the cosmological horizon for an observer moving along 
the $r=0$ geodesic.

We can now transform the coordinates to static slicing
\begin{equation}
  ds^2 = -\biggl( 1 - \frac{\rho^2}{\alpha^2} \biggr) d\tau^2 
    + \frac{1}{1 - \frac{\rho^2}{\alpha^2}} d\rho^2 
    + \rho^2 d\Omega_{d-1}^2.
\label{eq:dS-static}
\end{equation}
In Fig.~\ref{fig:dS-static}, we depict constant $\tau$ (red) and constant 
$\rho$ (blue) slices, with the shaded region being covered by the 
coordinates.
\begin{figure}[t]
\begin{center}
  \includegraphics[height=6.5cm]{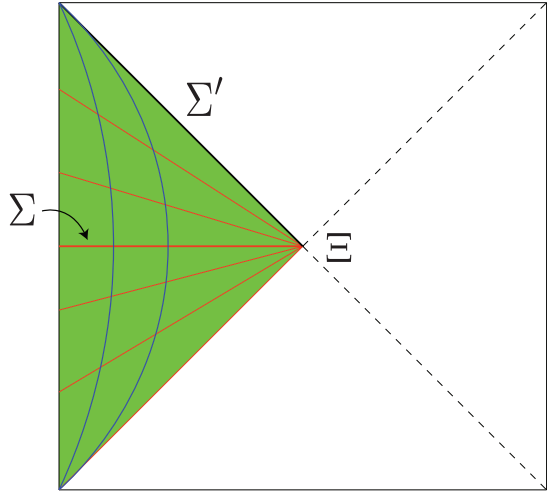}
\end{center}
\caption{Constant time slices and the spacetime region covered by 
 the coordinates in static slicing of de~Sitter space.  Here, $\Sigma$ 
 is the $\tau = 0$ hypersurface, and $\Xi$ is the bifurcation surface, 
 given by $\rho = \alpha$ with finite $\tau$.}
\label{fig:dS-static}
\end{figure}
This metric makes it manifest that the spacetime has a Killing symmetry 
corresponding to $\tau$ translation.  Using this symmetry, we can map 
a leaf of the original FRW universe to the $\tau = 0$ hypersurface, 
$\Sigma$.  Since the leaf of the universe under consideration approaches 
arbitrarily close to Eq.~(\ref{eq:r_dS}) at late times, the image of 
the map, $\Xi'$, asymptotes to the bifurcation surface $\Xi$ at
\begin{equation}
  \rho = \alpha,
\label{eq:rho_dS}
\end{equation}
for a leaf at later times.

Consider an arbitrary subregion $A$ on $\Xi'$ and the minimal area 
surface $\gamma_A$ on $\Sigma$ anchored to the boundary of $A$, 
$\partial A$.  Since the geometry of $\Sigma$ is $S^d$ with $\Xi$ 
being an equator, the minimal area surface $\gamma_A$ becomes the 
region $A$ itself (or its complement on $\Xi'$, whichever is smaller) 
in the limit $\Xi' \rightarrow \Xi$.  Strictly speaking, this statement 
does not apply for a small subset of subregions, since $\Xi'$ is 
not exactly $\Xi$ unless the leaf under consideration is at strictly 
infinite time.  (For subregions in this subset, the minimal area 
surfaces probe $\rho \ll \alpha$.  For spherical caps, these subregions 
are approximately hemispheres.)  However, the fractional size of the 
subset goes to zero as we focus on later leaves.  Continuity then 
tells us that our conclusion persists for all subregions.

The surface $\gamma_A$ found above is in fact an extremal surface, 
since the bifurcation surface $\Xi$ is an extremal surface, so any 
subregion of it is also extremal.  It is easy to show that this surface 
is indeed the HRRT surface, the minimal area extremal surface.  Suppose 
there is another extremal surface $\gamma'_A$ anchored to $\partial A$. 
We could then send a null congruence from $\gamma'_A$ down to $\Sigma$, 
yielding another codimension-2 surface $\gamma_A''$ given by the 
intersection of the null congruence and $\Sigma$.  Because $\gamma'_A$ 
is extremal, the focusing of the null rays implies $\norm{\gamma'_A} 
> \norm{\gamma''_A}$, and by construction $\norm{\gamma_A} < 
\norm{\gamma''_A}$.  This implies that $\gamma_A$ is the HRRT surface, 
and hence
\begin{equation}
  S_A = \frac{1}{4 l_{\rm P}^{d-1}} 
    {\rm min}\{ \norm{A}, \norm{\bar{A}} \}.
\label{eq:S_A-norm_A}
\end{equation}
Namely, the holographic state representing an FRW universe that 
asymptotically approaches de~Sitter space becomes a maximally entropic 
state in the late time limit.

The global spacetime structure in the case of a single fluid component 
with $w \neq -1$ is qualitatively different from the case discussed 
above.  For example, the area of a leaf grows indefinitely.  However, 
for any finite time interval, the behavior of the system approaches that 
of de~Sitter space in the limit $w \rightarrow -1$.  In fact, the numerical 
analysis of Ref.~\cite{Nomura:2016ikr} tells us that the holographic 
entanglement entropy of a spherical cap region becomes maximal in the 
$w \rightarrow -1$ limit.  We show in Appendix~\ref{subapp:FRWlimit} 
that this occurs for an arbitrary subregion on a leaf.

\subsection{Spacetime disappears as {\boldmath $w \rightarrow -1$} in the 
holographic FRW theory}
\label{subsec:dS_no-spacetime}

We have seen in our AdS/CFT example that as the holographic state 
approaches typicality, and hence becomes maximally entropic, the 
directly reconstructable region disappears.  On the other hand, we 
have shown that the entanglement entropies for flat FRW universes 
approaches the maximal form as $w \rightarrow -1$.  Does this limit 
have a corresponding disappearance of reconstructable spacetime? 
Here we will show that the answer to this question is yes.

From the analysis of Section~\ref{subsec:dS_max-ent}, we see that a 
leaf at late times in universes approaching de~Sitter space can be mapped 
to a surface on the $\tau = 0$ hypersurface $\Sigma$, which asymptotes 
to the bifurcation surface $\Xi$ in the late time limit.  From the 
Killing symmetry, the HRRT surfaces anchored to this mapped leaf must 
all be restricted to living on $\Sigma$.  Mapping the HRRT surfaces 
back to the original location, we see that they asymptote to living 
on the null hypersurface $\Sigma'$.  Thus, we find that the HRRT 
surface for any subregion of a leaf $\sigma_*$ asymptote to the 
future boundary of the causal region $D_{\sigma_*}$, which we denote 
by $\partial D_{\sigma_*}^{(+)}$, as a universe approaches de~Sitter 
space.  A similar argument holds for universes where $w \rightarrow -1$. 
In Appendix~\ref{subapp:HRRT-dS}, we present some examples where we 
can see this behavior using analytic expressions for HRRT surfaces.

What does this imply for the reconstructable region in de~Sitter space? 
Using the prescription outlined in Section~\ref{subsec:reconst}, we 
find that spacetime points on the future causal boundary of a leaf, 
$\partial D_{\sigma_*}^{(+)}$, can be reconstructed.  This is a 
codimension-1 region in spacetime.  One might then think that we 
can reconstruct a codimension-0 region by considering multiple leaves, 
as was the case in a Schwarzschild-AdS black hole.  However, the 
holographic screen of de~Sitter space is itself a null hypersurface, 
with future leaves lying precisely on the future causal boundary of 
past leaves.  This means that even by using multiple leaves we cannot 
reconstruct any nonzero measure spacetime region in the de~Sitter 
(and $w \rightarrow -1$) limit.

We will now compute the reconstructable region in $(2+1)$-dimensional flat 
FRW spacetimes.  As discussed in Section~\ref{subsec:reconst}, this region 
is comprised of points that can be localized as the intersection of 
edges of entanglement wedges.  We will be considering the reconstructable 
region associated to a single leaf, and hence this prescription reduces 
to finding points located at the intersection of HRRT surfaces anchored 
to the leaf.  This alone gives us a codimension-0 reconstructable region. 
In $(2+1)$-dimensional FRW spacetimes, HRRT surfaces are simply geodesics 
in the bulk spacetime, and this problem becomes tractable.

For a $(2+1)$-dimensional flat FRW universe filled with a single fluid 
component $w$, the leaf of the holographic screen at conformal time 
$\eta_*$ is located at coordinate radius
\begin{equation}
  r_* = \frac{a}{\dot{a}}\biggr|_{\eta = \eta_*} = w \eta_*.
\label{eq:r_*}
\end{equation}
Let us parameterize the points on the leaf by $\phi \in [0, 2 \pi)$. 
Consider an interval of the leaf at time $\eta_*$ centered at $\phi_0$ 
with half opening angle $\psi$.  The HRRT surface of this subregion 
is simply the geodesic connecting the endpoints of the interval:\ 
$(\eta, \phi) = (\eta_*, \phi_0 - \psi)$ and $(\eta_*, \phi_0 + \psi)$. 
It is clear from the symmetry of the setup that if we consider a second 
geodesic anchored to an interval with the same opening angle but with 
a center $\phi'_0 \in [\phi_0 - 2\psi, \phi_0 + 2\psi]$, then the 
two geodesics will intersect at a point, specifically where $\phi 
= (\phi_0 + \phi'_0)/2$.  Using these pairs of geodesics, it is clear 
that we can reconstruct all points on all geodesics anchored to the 
leaf.  The union of these points gives us a codimension-0 region.

Can we get a larger region?  In $(2+1)$-dimensional flat FRW spacetimes, 
the answer is no.  In higher dimensions, knowing the HRRT surfaces 
for all spherical cap regions may not be sufficient to figure out 
reconstructable regions; for example, one may consider using disjoint 
regions in hopes that the new HRRT surfaces would explore regions 
inaccessible to the previous HRRT surfaces (although we do not know 
if this really leads to a larger reconstructable region).  However, in 
$2+1$ dimensions, both connected and disconnected phases of extremal 
surfaces are constructed from the geodesics already considered, so we 
gain nothing from considering disconnected subregions.  We thus find 
that the set of all points on HRRT surfaces anchored to arbitrary 
subregions on a leaf is exactly the reconstructable region from the 
state on the leaf.

\begin{figure}[t]
\begin{center}
  \includegraphics[height=6.5cm]{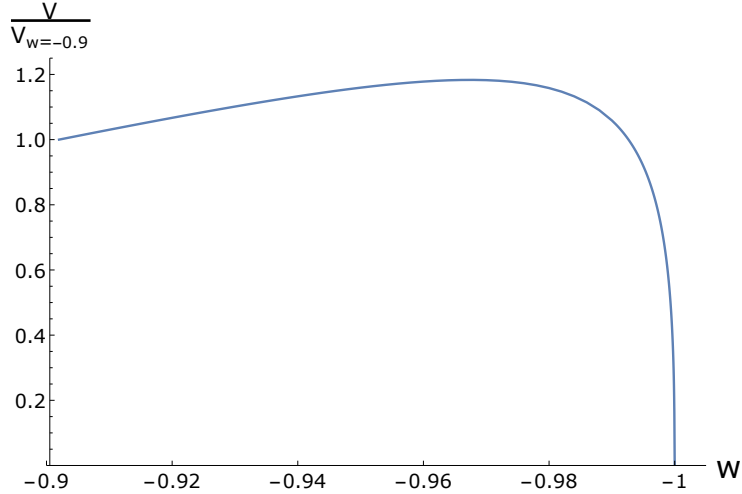}
\end{center}
\caption{The spacetime volume of the reconstructable region in 
 $(2+1)$-dimensional flat FRW universes for $w \in (-0.9, -1)$, 
 normalized by the reconstructable volume for $w = -0.9$.}
\label{fig:2dVol}
\end{figure}
In Fig.~\ref{fig:2dVol}, we show a plot of the reconstructable spacetime 
volume as a function of $w$.  It shows a qualitatively similar behavior 
to that of Fig.~\ref{fig:f}, where the reconstructable volume increases 
and then sharply declines to zero as the holographic state becomes 
maximally entropic.

We can also perform a similar analysis in higher dimensions.  Due to 
the numerical difficulty in finding extremal surfaces, here we restrict 
ourselves to the region reconstructable by spherical cap regions 
(which may indeed be the fully reconstructable region) and to only 
a few representative values of $w$.  The results are plotted in 
Fig.~\ref{fig:recon3D} for $(3+1)$-dimensional FRW universes.  These 
demonstrate the behavior that the extremal surfaces, and hence the 
reconstructable region, becomes more and more null as $w \rightarrow -1$.
\begin{figure}
\hfill
\subfigure[$w = 1$, $\eta_* = 1$]
  {\includegraphics[width=.3 \linewidth]{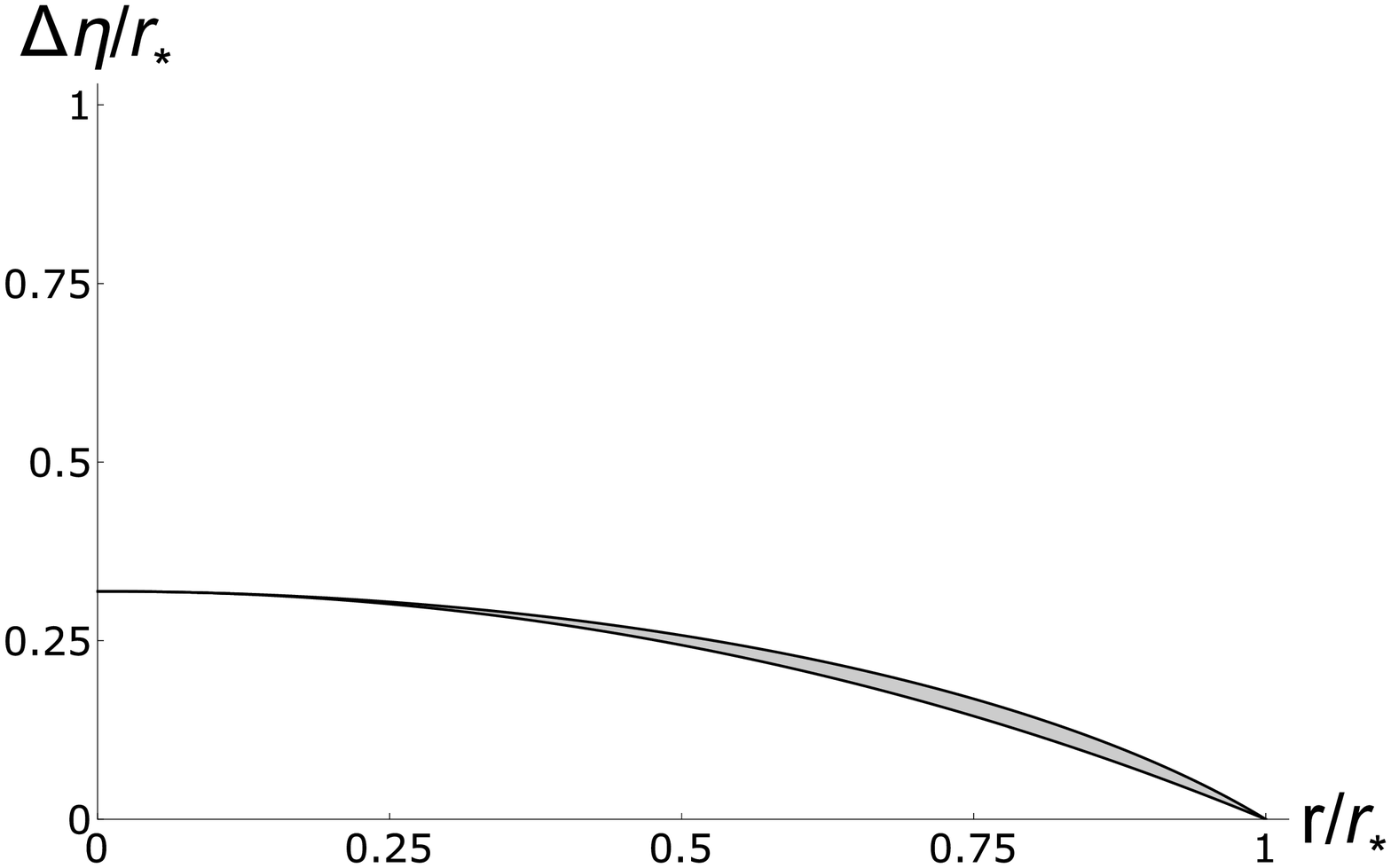}}
\hfill
\subfigure[$w = 0$, $\eta_* = 1$]
  {\includegraphics[width=.3 \linewidth]{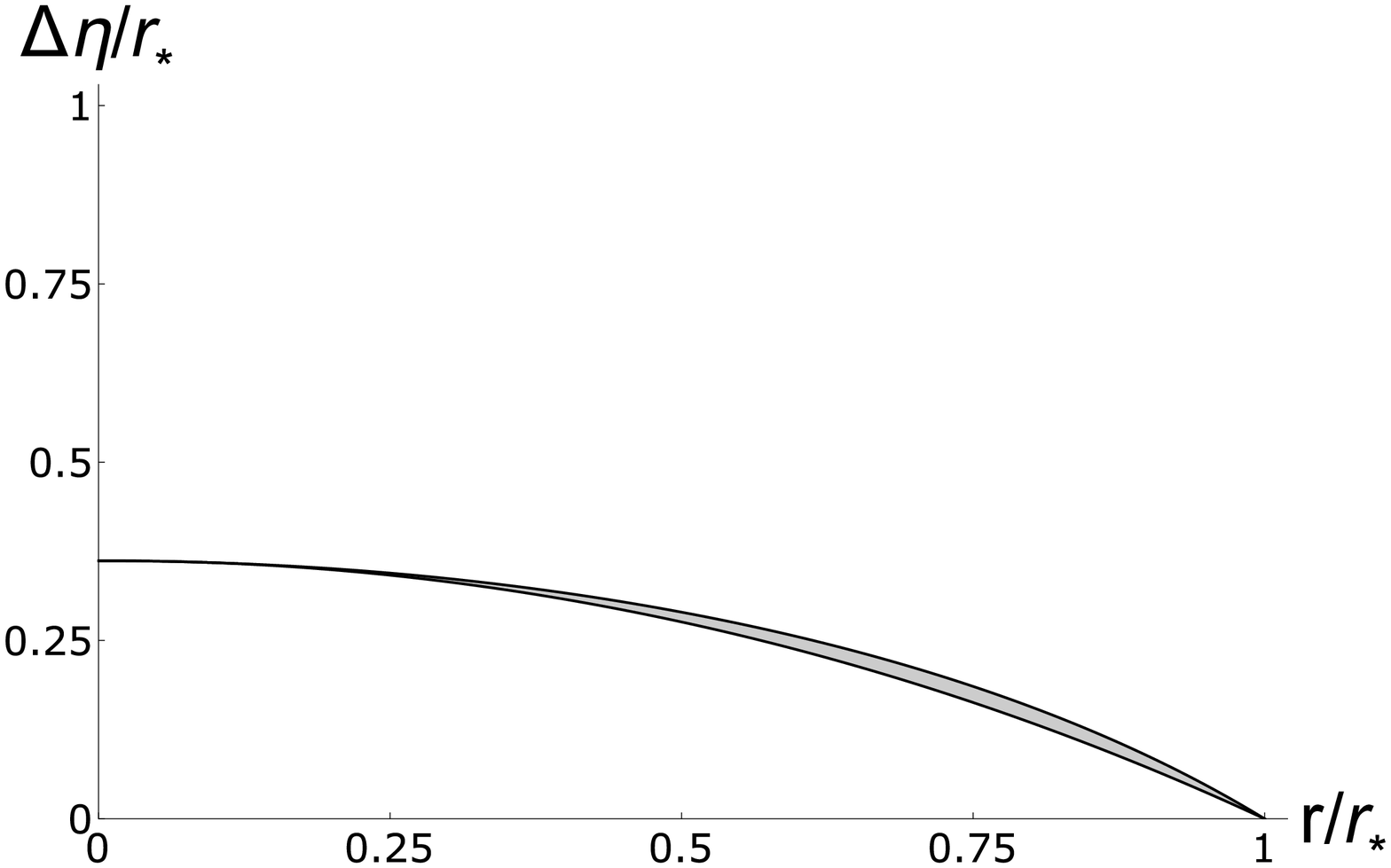}}
\hfill
\subfigure[$w = -4/5$, $\eta_* = -1$]
  {\includegraphics[width=.3 \linewidth]{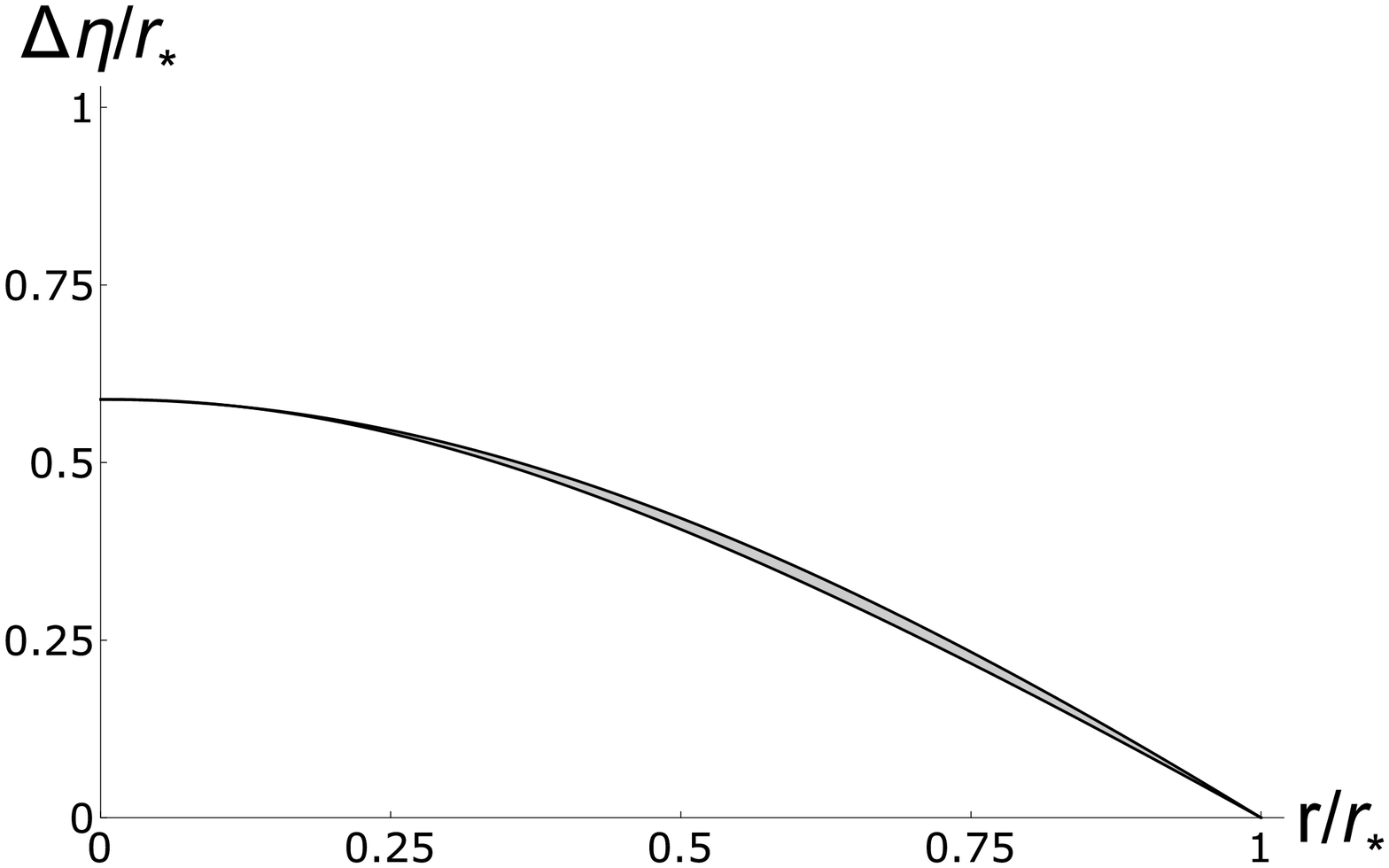}}
\caption{Reconstructable spacetime regions for various values of $w$ 
 in $(3+1)$-dimensional flat FRW universes.  The horizontal axis is 
 the distance from the center, normalized by that to the leaf.  The 
 vertical axis is the difference in conformal time from the leaf, 
 normalized such that null ray from the leaf would reach $1$.  The 
 full reconstructable region for each leaf would be the gray region 
 between the two lines rotated about the vertical axis.}
\label{fig:recon3D}
\end{figure}

The discussion in this subsection says that the reconstructable 
spacetime region disappears in the holographic theory of FRW spacetimes 
as the holographic state becomes maximally entropic in the de~Sitter 
limit.  While a microstate becoming maximally entropic does not directly 
mean that states representing the corresponding spacetime become typical 
in the holographic Hilbert space (since the number of independent 
microstates could still be small), we expect that the former indeed 
implies the latter as usual thermodynamic intuition suggests; 
see Section~\ref{sec:linearity} for further discussion.  In any 
event, since typical states in a holographic theory are maximally 
entropic, we expect that the reconstructable spacetime region disappears 
as the holographic state becomes typical.

An important implication of the analysis here is that a holographic 
theory of de~Sitter space cannot be obtained by taking a limit in 
the holographic theory of FRW spacetimes.  A holographic theory 
of exact de~Sitter space, if any, would have to be formulated in 
a different manner.%
\footnote{Another instance in which spacetime disappears is when the 
 holographic description changes from that based on a past holographic 
 screen (foliated by marginally anti-trapped surfaces) to a future 
 holographic screen (marginally trapped surfaces).  Such a change 
 of description may occur in a spacetime with a late-time collapsing 
 phase, e.g.\ in a closed FRW universe with the holographic screen 
 constructed naturally in an observer-centric manner.  (For an 
 interpretation of such spacetime, see Ref.~\cite{Nomura:2016ikr}.) 
 Since the leaf at the time of the transition is extremal, the analysis 
 here indicates that the spacetime region reconstructable from a single 
 leaf disappears at that time.  This makes the necessity of the change 
 of the description more natural. \label{ft:collapse}}

\subsection{Maximally entropic states have no spacetime}
\label{subsec:proof}

In this subsection, we provide a proof for the statement that the 
directly reconstructable region of a maximally entropic leaf is either 
the leaf itself or a subset of its null cone.  We use this result to 
argue that maximally entropic states have no spacetime.  This heavily 
utilizes the maximin techniques developed in Ref.~\cite{Wall:2012uf}.

\begin{theorem}
Consider a compact codimension-2 spacelike surface, $\sigma$, 
with area ${\cal A}$, living in a spacetime that satisfies $R_{ab} 
v^a v^b \ge 0$ for all null vectors $v^a$.  Suppose HRRT surfaces 
can consistently be anchored to $\sigma$.%
\footnote{This requires the expansion of the two null hypersurfaces 
 bounding $D(\sigma)$ to have $\theta \leq 0$, where $D(\sigma)$ 
 is the interior domain of dependence of some achronal set 
 whose boundary is $\sigma$.  These HRRT surfaces are guaranteed 
 to exist and satisfy basic entanglement inequalities; see 
 Refs.~\cite{Sanches:2016sxy,Miyaji:2015yva}.}
Let $m(\Gamma)$ denote the HRRT surface anchored to the boundary, 
$\partial \Gamma$, of a subregion $\Gamma$ of $\sigma$.

If $\norm{m(\Gamma)} = {\rm min}\{ \norm{\Gamma},\norm{\bar{\Gamma}} \}, 
\forall \Gamma \subset \sigma$, then either $\sigma$ is a bifurcation 
surface or all of the HRRT surfaces of $\sigma$ lie on a non-expanding 
null hypersurface connected to $\sigma$.
\label{th:1}
\end{theorem}
 
\begin{proof}
If $\Gamma_1$ and $\Gamma_2$ are subregions of $\sigma$, we will 
abbreviate $\Gamma_1 \cup \Gamma_2$ as $\Gamma_1 \Gamma_2$.  Let 
$m(\Gamma)_\Sigma$ denote the representative of $m(\Gamma)$ on a 
complete achronal surface $\Sigma$, defined by the intersection of 
$\Sigma$ with a null congruence shot out from $m(\Gamma)$.  From 
the extremality of $m(\Gamma)$, $R_{ab} v^a v^b \ge 0$, and the 
Raychaudhuri equation, $\norm{m(\Gamma)_\Sigma} \leq \norm{m(\Gamma)}$.

\begin{figure}[t]
\begin{center}
  \includegraphics[height=7cm]{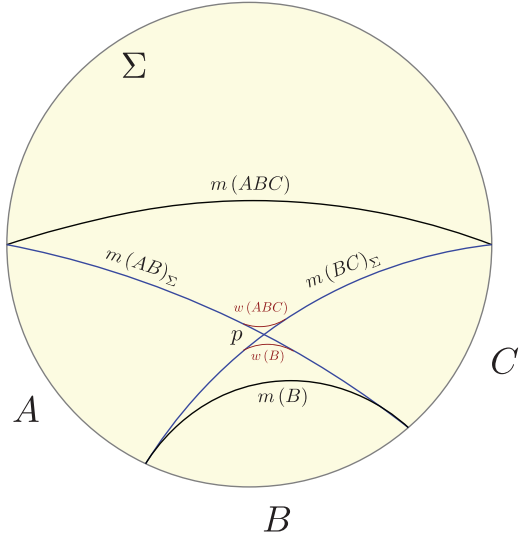}
\end{center}
\caption{Diagrams representing the achronal surface $\Sigma$ in 
 which two HRRT surfaces, $m(ABC)$ and $m(B)$, live.  $m(AB)_\Sigma$ 
 and $m(BC)_\Sigma$ are the representatives of $m(AB)$ and $m(BC)$, 
 respectively.  They are shown to be intersecting at $p$.  On a 
 spacelike $\Sigma$, one could deform around this intersection to 
 create two new surfaces with smaller areas.}
\label{fig:proof1}
\end{figure}
Consider three connected subregions $A$, $B$, $C$ of $\sigma$ such 
that $\partial A \cap \partial B \neq \emptyset$, $\partial B \cap 
\partial C \neq \emptyset$ where both such intersections are 
codimension-3, and $\norm{A \cup B \cup C} \leq \norm{\sigma}/2$; 
see Fig.~\ref{fig:proof1} for a diagram.  By Theorem~17.h of 
Ref.~\cite{Wall:2012uf}, take $m(ABC)$ and $m(B)$ to be on the 
same achronal surface, $\Sigma$.  Now, consider the representatives 
$m(AB)_\Sigma$ and $m(BC)_\Sigma$.  From the properties of 
representatives and maximin surfaces, we have
\begin{equation}
  S(AB) + S(BC) 
  \ge \frac{\norm{m(AB)_\Sigma}}{4 l_{\rm P}^{d-1}} 
    + \frac{\norm{m(BC)_\Sigma}}{4 l_{\rm P}^{d-1}} 
  \ge S(ABC) + S(B).
\label{eq:proof-ineq}
\end{equation}
The assumption of maximal entropies then tells us that strong 
subadditivity is saturated, and hence
\begin{equation}
\begin{aligned}
\norm{m(AB)_\Sigma} &= \norm{m(AB)}, \\
\norm{m(BC)_\Sigma} &= \norm{m(BC)}.
\end{aligned}
\label{eq:rep-eq}
\end{equation} 
Additionally, $m(AB)_\Sigma \cap m(BC)_\Sigma \neq \emptyset$.

We have two cases depending on the nature of $\Sigma$.
\vspace{0.2cm}

\noindent
Case~1: $m(ABC)$, $m(B)$, $m(AB)_\Sigma$, and $m(BC)_\Sigma$ live on 
$\Sigma$ which is a non-null hypersurface.

\begin{itemize}
\item[]
Suppose $m(AB)_\Sigma \cap m(BC)_\Sigma$ is a codimension-3 
surface, meaning they intersect through some surface, $p$, depicted 
in Fig.~\ref{fig:proof1}.  One could then smooth out the ``corners'' 
around $p$ to create new surfaces homologous to $ABC$ and $B$.  This 
is depicted through the maroon lines in Fig.~\ref{fig:proof1}.  By the 
triangle inequality, these new, smoothed out surfaces, $w(ABC)$ and 
$w(B)$, would have less total area than $m(ABC) \cup m(B)$ because 
$p \in \Sigma$, which is spacelike.  However, this contradicts the 
minimality of $m(ABC)$ and $m(B)$:
\begin{align}
  & (\norm{A} + \norm{B} + \norm{C}) + \norm{B} 
  = \norm{m(ABC)} + \norm{m(B)} 
  \leq \norm{w(ABC)} + \norm{w(B)} 
\nonumber\\
  & \qquad < \norm{m(AB)_\Sigma} + \norm{m(BC)_\Sigma} 
  = (\norm{A} + \norm{B}) + (\norm{B} + \norm{C}).
\label{eq:contrad}
\end{align}
Therefore, $m(AB)_\Sigma$ and $m(BC)_\Sigma$ cannot intersect through 
some codimension-3 surface, yet they must still intersect.  This requires 
$m(AB)_\Sigma$ and $m(BC)_\Sigma$ to coincide somewhere, a neighborhood 
of $x$, and by Theorem~4.e of Ref.~\cite{Wall:2012uf} these two 
surfaces must coincide at every point connected to $x$.  This means 
that $m(AB) = m(A) \cup m(B)$ and $m(BC) = m(B) \cup m(C)$.  The only 
way this can consistently occur for all possible $A$, $B$, and $C$ is 
for $m(\Gamma) \subset \sigma$.  This means that $\sigma$ itself is 
extremal, and hence is a bifurcation surface.
\end{itemize}

\noindent
Case~2: $m(ABC)$, $m(B)$, $m(AB)_\Sigma$, and $m(BC)_\Sigma$ live on 
hypersurface $\Sigma$ which is at least partially null.

\begin{itemize}
\item[]
Suppose $m(AB)_\Sigma \cap m(BC)_\Sigma$ is a codimension-3 surface, 
$p$.  If at $p$, $\Sigma$ is null and non-expanding, then smoothing 
out the intersection will not result in new surfaces with smaller area. 
This is the condition that $\theta_u = 0$ on the null hypersurface, 
$\Sigma_u$, coincident with $\Sigma$ at $p$, where $u$ is the null 
vector generating $\Sigma_u$ at $p$.  Hence, the representatives can 
intersect at $p$ and simultaneously saturate strong subadditivity. 
Additionally, because $\norm{m(BC)_\Sigma} = \norm{m(BC)}$, we know 
that $\theta_v = 0$ along the hypersurface generating the representatives 
of $m(BC)$, where $v$ is the null vector generating this hypersurface. 
Therefore at $p$, $\theta_u = \theta_v = 0$.

\begin{figure}[t]
\begin{center}
  \includegraphics[height=6.6cm]{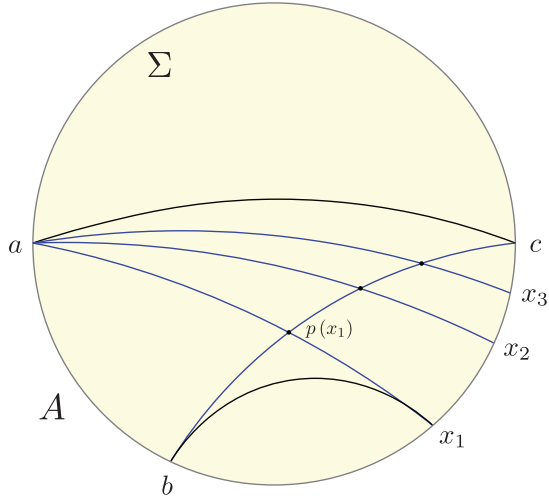}
\end{center}
\caption{This depicts how one can scan across the representative 
 $m(BC)_\Sigma$ by bipartitioning $BC$ on the achronal surface $\Sigma$. 
 At each of these intersections, $p(x_i)$, $\theta_u = \theta_v = 0$ 
 if the state on the leaf is maximally entropic and $\Sigma$ is null 
 and non-expanding.}
\label{fig:proof2}
\end{figure}
We can now scan across $m(BC)_\Sigma$ by considering its intersection 
with $m(AB')_\Sigma$ where $(B', C')$ is a bipartition of $B \cup C$ 
where $\partial{B'} \cap \partial{A} \neq \emptyset$, and then 
considering all such bipartitions.  This is illustrated in 
Fig.~\ref{fig:proof2} by splitting up $B \cup C$ at a few points 
labeled by $x_i$; for example, $B' = [b, x_3]$ and $C' = [x_3, c]$ 
is one such allowable bipartition.  By continuity, all of 
$m(BC)_\Sigma$ will be scanned.%
\footnote{We believe this is sufficient to scan over the whole surface 
 assuming the spacetime is smooth.  Additionally, Eq.~(\ref{eq:rep-eq}) 
 requires there to be no energy density between an HRRT surface and 
 its representative, this will preclude jumps in the representatives 
 due to entanglement shadows and the like.}

By the argument in the previous paragraph, all intersection points 
along $m(BC)_\Sigma$ must then have $\theta_u = \theta_v = 0$.  Assuming 
nondegeneracy, $m(BC)_\Sigma$ must therefore be the HRRT surface 
$m(BC)$.  Additionally, every point of $m(BC)$ lives on some null, 
non-expanding hypersurface and at $\partial m(BC)$ this surface connects 
to $\sigma$.  Hence, at $\partial m(BC)$, $\sigma$ must be marginal.
This argument can be repeated for any set of appropriate subregions. 
This tells us that all HRRT surfaces have the previously stated 
properties.

Now, by Theorem~17.h of Ref.~\cite{Wall:2012uf}, we can construct an 
achronal surface, $\Sigma$, that is foliated by HRRT surfaces.  Each 
point of $\Sigma$ must now be null and non-expanding.  Additionally, 
the boundary of $\Sigma$, $\sigma$, must be marginal.  Let $k$ denote 
the vector in this local marginal direction.  This uniquely specifies 
$\Sigma$ as the null non-expanding hypersurface generated by $k$. 
This is true for all $\Sigma$ foliated by HRRT surfaces, and each 
HRRT surface can belong to some foliation of a $\Sigma$.%
\footnote{Under the assumption of the theorem, the HRRT surface of 
 disconnected subregions will always be disconnected.  This is because 
 the disconnected surface is extremal.}
Hence all extremal surfaces anchored to $\sigma$ must belong to a 
non-expanding null hypersurface.

Back to the beginning, if the intersection of $m(AB)_\Sigma$ and 
$m(BC)_\Sigma$ is codimension-2, then the argument from Case~1 applies 
and $\sigma$ must be extremal.
\end{itemize}
This concludes the proof of Theorem~\ref{th:1}.
\end{proof}

\begin{corollary}
Consider a codimension-2 surface, $\sigma$, with area ${\cal A}$, 
living in a spacetime satisfying $R_{ab} v^a v^b \ge 0$.  Let $m(\Gamma)$ 
denote the HRRT surface anchored to $\partial \Gamma$.

If $\sigma$ is not marginal, then it cannot satisfy $\norm{m(\Gamma)} 
= {\rm min}\{ \norm{\Gamma},\norm{\bar{\Gamma}} \}, \forall \Gamma 
\subset \sigma$.
\label{corr:1}
\end{corollary}

\begin{proof}
The contrapositive of this statement is proven by Theorem~\ref{th:1}.
\end{proof}

Consider the case that $\sigma$ is a leaf of a past holographic screen. 
If the leaf is extremal and the screen is not null, then the directly 
reconstructable spacetime is just the leaf itself.  Additionally, this 
tells us the holographic screen must halt at this point.  This indicates 
the end of a holographic description based on the past holographic screen. 
At this point, one can stitch the beginning of a new future holographic 
screen that starts at a bifurcation surface, patching together two 
holographic descriptions.  This occurs in collapsing universes; see 
footnote~\ref{ft:collapse}.

In the other case, if all of the HRRT surfaces of $\sigma$ have area 
corresponding to the maximal entropy, then all of the extremal surfaces 
must lie on the future null cone of the leaf, where this null cone is 
non-expanding and compact.  This cone itself is the limit of a past 
holographic screen because $\theta_k = 0$.  Barring the existence of 
a continuum of compact, non-expanding, null hypersurfaces, the holographic 
screen then follows along this null surface from the leaf.  Hence the 
directly  reconstructable region will only be the screen itself, exactly 
as we observed in the case of de~Sitter space.  Again, we see that 
maximal entanglement corresponds to the end of a holographic description, 
but in this case the screen does not end; this corresponds to a stable 
final state.

In Section~\ref{subsec:AdS-BH}, we took the boundary to be at some 
large, fixed radius in AdS space.  One may be concerned that this 
cutoff surface is not marginal, and hence Theorem~\ref{th:1} does 
not apply.  However, in the limit that the black hole radius approaches 
the boundary, then the statement holds because the horizon of the black 
hole satisfies the needed properties.  Note that until this final limit, 
Corollary~\ref{corr:1} tells us that the entanglement of the boundary 
cannot be maximal.

Finally we are prepared to make a statement about typicality.  Typical 
boundary states are maximally entangled, and hence the argument shows 
us that for holographic theories living on screens (an instance of which 
is AdS/CFT), typical states have no directly reconstructable spacetime.

\section{Spacetime Emerges through Deviations from Maximal Entropy}
\label{sec:atypical}

We have seen that when the holographic state becomes maximally 
entropic, spacetime defined as the directly reconstructable region 
disappears.  Conversely, bulk spacetime emerges when we change 
parameters, e.g.\ the mass of the black hole or the equation of state 
parameter $w$, deviating the state from maximal entropy.  In this 
section, we study how this deviation may occur and find qualitative 
differences between the cases of Schwarzschild-AdS and flat FRW 
spacetimes.  This has important implications for the structures 
of holographic theories representing these spacetimes.

\subsection{CFT with subcutoff temperatures}
\label{subsec:lower-T}

Consider the setup discussed in Section~\ref{subsec:AdS-BH}:\ a large 
black hole in asymptotically AdS space.  The holographic theory is then 
a local quantum (conformal) field theory.  When the temperature of the 
system is at the cutoff scale, the holographic state has maximal entropies, 
Eq.~(\ref{eq:max-ent}).  As we lower the temperature, the state deviates 
from a maximally entropic one, and correspondingly bulk spacetime 
emerges---the horizon of the black hole recedes from the cutoff surface, 
and the reconstructable spacetime region appears; see Fig.~\ref{fig:f}.

Suppose the temperature of the system $T$ is lower than the cutoff scale, 
$T < \Lambda$.  We are interested in the behavior of von~Neumann entropies 
of subregions of characteristic length $L$ in the boundary theory.  These 
entropies are calculated holographically by finding the areas of the 
HRRT surfaces anchored to subregions of the cutoff surface $r = R$. 
We analyze this problem analytically for spherical cap regions in 
Appendix~\ref{subapp:HRRT_S-AdS}.  For sufficiently high temperature, 
$T \gg (\Lambda^{d-2}/l)^{1/(d-1)}$, we find that the entanglement 
entropy for a subregion $A$ behaves as
\begin{equation}
  S_A \approx 
  \left\{ \begin{array}{ll}
    c A_{d-2} L^{d-2} \Lambda^{d-2} &
    \mbox{for } L \ll L_*, \\
    c A_{d-2} \frac{r_+^{d-1} L^{d-1}}{l^{2d-2}} \approx 
      c \left(\frac{T}{\Lambda}\right)^{d-1} A_{d-2}L^{d-1}\Lambda^{d-1} &
    \mbox{for } L \gg L_*.
  \end{array} \right.
\label{eq:S_A-case}
\end{equation}
Here,
\begin{equation}
  L_* \approx \frac{l^2 R^{d-2}}{r_+^{d-1}} 
  \approx \frac{\Lambda^{d-2}}{T^{d-1}},
\label{eq:L_tr}
\end{equation}
$c \approx (l/l_{\rm P})^{d-1}$ is the central charge of the CFT, 
and $A_{d-2}$ is the area of the $(d-2)$-dimensional unit sphere. 
We find that the scaling of the entanglement entropy changes (smoothly) 
from an area law to a volume law as $L$ increases.  For $T \ll 
(\Lambda^{d-2}/l)^{1/(d-1)}$, i.e.\ $L_* \gg l$, the entanglement 
entropy obeys an area law for all subregions.  We note that the length 
in the boundary theory is still measured in terms of the $d$-dimensional 
metric at infinity with the conformal factor stripped off.  The cutoff 
length is thus $1/\Lambda \approx O(l^2/R)$, and the size of the boundary 
space is $\approx O(l)$.

While we have analyzed spherical cap subregions, the behavior of the 
entanglement entropy found above is more general.  When the temperature 
is lowered from the cutoff scale, the entanglement entropy $S_A$ 
deviates from the maximal value.  Defining
\begin{equation}
  Q_A = \frac{S_A}{S_{A,{\rm max}}} 
  = \frac{S_A}{\norm{A}/4l_{\rm P}^{d-1}},
\label{eq:Q_A-def}
\end{equation}
we find that
\begin{equation}
  Q_A \approx 
  \left\{ \begin{array}{ll}
    \frac{1}{L \Lambda} &
    \mbox{for } L \ll \frac{\Lambda^{d-2}}{T^{d-1}}, \\
    \left(\frac{T}{\Lambda}\right)^{d-1} &
    \mbox{for } L \gg \frac{\Lambda^{d-2}}{T^{d-1}}.
  \end{array} \right.
\label{eq:Q_A-BH}
\end{equation}
Here, we have assumed that subregion $A$ is characterized by a single 
length scale $L$, and that the temperature is sufficiently high, $T \gg 
(\Lambda^{d-2}/l)^{1/(d-1)}$.  (If $T \ll (\Lambda^{d-2}/l)^{1/(d-1)}$, 
$Q_A \approx 1/L\Lambda$ for all subregions.)  This behavior is depicted 
schematically in Fig.~\ref{fig:Q_A-AdS}.
\begin{figure}[t]
\begin{center}
  \includegraphics[height=6.5cm]{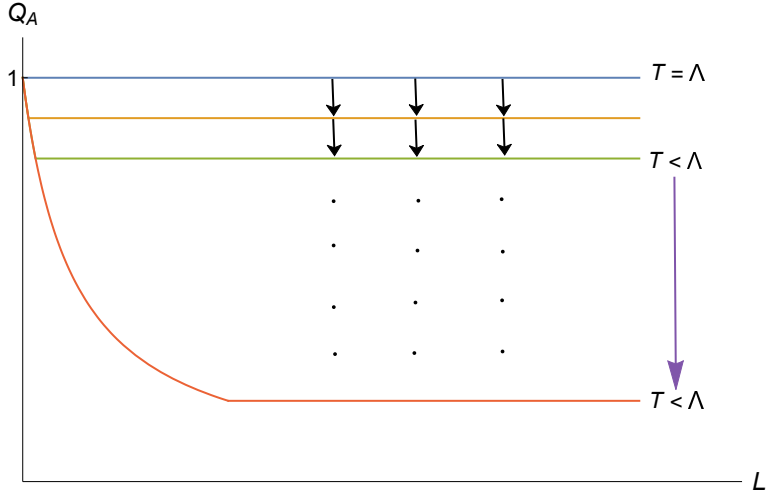}
\end{center}
\caption{A schematic depiction of the entanglement entropy in the 
 Schwarzschild-AdS spacetime, normalized by the maximal value of entropy 
 in the subregion, $Q_A = S_A/S_{A,{\rm max}}$, and depicted as a 
function of the size $L$ of subregion $A$; see Eq.~(\ref{eq:Q_A-BH}). 
 The scales of the axes are arbitrary.  As the mass of the black hole 
 is lowered (the temperature $T$ of the holographic theory is reduced 
 from the cutoff $\Lambda$), $Q_A$ deviates from $1$ in a specific 
 manner.}
\label{fig:Q_A-AdS}
\end{figure}

We find that as the temperature is lowered from the cutoff scale, two 
things occur for entanglement entropies:
\begin{itemize}
\item
For sufficiently large subregions, the entanglement entropies still 
obey a volume law, but the coefficient becomes smaller.
\item
The more the temperature is lowered, the further subregions have 
entanglement entropies obeying an area law.  This occurs from shorter 
scales, i.e.\ subregions with smaller sizes.
\end{itemize}
These make the entanglement entropies deviate from the maximal value 
and lead to the emergence of reconstructable spacetime:\ the region 
between the black hole horizon and the cutoff surface, $r_+ < r \leq R$.

\subsection{FRW universes with {\boldmath $w > -1$}}
\label{subsec:FRW-Q}

As spacetime emerges by reducing the mass of the black hole in the 
Schwarzschild-AdS case, a codimension-0 spacetime region that is 
reconstructable from a single leaf appears when $w$ is increased 
from $-1$.  As in the AdS case, this appearance is associated 
with a deviation of entanglement entropies from saturation.  However, 
the manner in which this deviation occurs is qualitatively different 
in the two cases.

To illustrate the salient points, let us consider flat FRW spacetimes 
with a single fluid component $w$ and a spherical cap region $A$ on 
a leaf parameterized by the half opening angle $\psi$.  Below, we 
focus on entanglement entropies $S_w(\psi)$ of the regions with 
$\psi \leq \pi/2$.  Those with $\psi > \pi/2$ are given by the 
relation $S_w(\psi) = S_w(\pi-\psi)$.

As before, we define
\begin{equation}
  Q_w(\psi) = \frac{S_w(\psi)}{S_{\rm max}(\psi)} 
  = \frac{S_w(\psi)}{\norm{A}/4l_{\rm P}^{d-1}}.
\label{eq:Q_w-psi}
\end{equation}
This quantity was calculated in Ref.~\cite{Nomura:2016ikr} in $3+1$ 
dimensions, which we reproduce in Fig.~\ref{fig:Q_w-FRW}.
\begin{figure}[t]
\begin{center}
  \includegraphics[height=6.5cm]{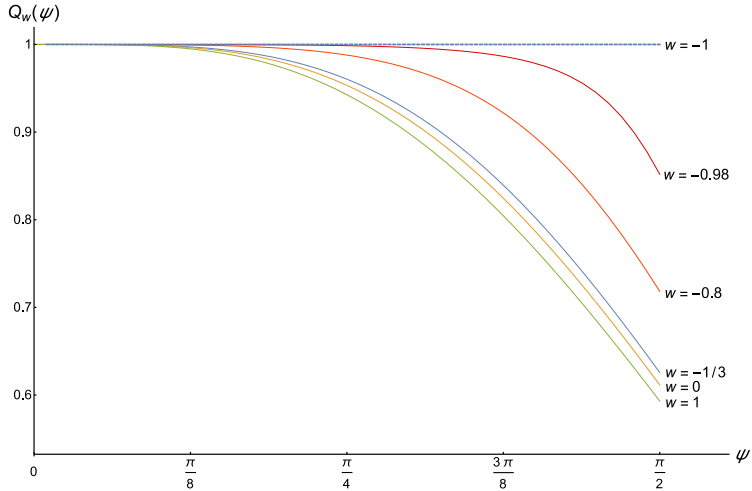}
\end{center}
\caption{The entanglement entropy in the holographic theory of flat 
 FRW spacetimes normalized by the maximal value of entropy in the 
 subregion, $Q_w(\psi) = S_w(\psi)/S_{\rm max}(\psi)$, as a function 
 of the size of the subregion, a half opening angle $\psi$.  As the 
 equation of state parameter $w$ is increased from $-1$, $Q_w(\psi)$ 
 deviates from $1$ in a way different from the Schwarzschild-AdS case.}
\label{fig:Q_w-FRW}
\end{figure}
The basic features are similar in other dimensions. 
In particular, $Q_w(\psi)$ satisfies the properties 
given in Eqs.~(\ref{eq:prop-Q-1},~\ref{eq:prop-Q-2}) in 
Appendix~\ref{subapp:2Dproof}.

We find that the way $Q_w(\psi)$ deviates from $1$ as $w$ is increased 
from $-1$ is qualitatively different from the way the similar quantity 
$Q_A$ deviates from $1$ in the Schwarzschild-AdS case as the temperature 
is reduced from the cutoff scale.  In particular, we find that in 
the FRW case
\begin{itemize}
\item
The deviation from $Q_w(\psi) = 1$ occurs {\it from larger subregions}. 
Namely, as $w$ is raised from $-1$, $Q_w(\psi)$ is reduced from $1$ 
first in the vicinity of $\psi = \pi/2$.
\item
There is no regime in which the entanglement entropy obeys an area law, 
$Q_w(\psi) \sim 1/\psi$, or a volume law with a reduced coefficient, 
$Q_w(\psi) = {\rm const.} < 1$.
\end{itemize}
As we will see next, these have profound implications for the nature of 
the holographic theory of FRW spacetimes.

\subsection{Locality vs nonlocality}
\label{subsec:locality}

In the following discussion, we assume that the dynamics in the 
holographic theory are chaotic and non-integrable as expected in 
a theory of quantum gravity; see, e.g. Ref.~\cite{Maldacena:2015waa}. 
Such systems are expected to satisfy the eigenstate thermalization 
hypothesis (ETH)~\cite{Srednicki:1994,Nandkishore:2014kca}, so generic 
high energy eigenstates reproduce the behavior of a thermal Gibbs 
density matrix.  In addition, we note that the dimension of the 
holographic Hilbert space is large $({\cal A}/4 l_{\rm P}^{d-1} \gg 1)$ 
and finite size effects causing deviations from the thermodynamic 
limit can be ignored.

We have already seen that one way to obtain a maximally entropic state 
is to look at high energy states in a local theory.  In the context of 
AdS/CFT, this corresponds to examining black holes with temperature near 
the cutoff scale.  To deviate from maximal entropy, one can then simply 
lower the energy of the states being considered.  For subregions beyond 
the correlation length, the reduced density matrix is well approximated 
by a Gibbs density matrix, and hence the entropy obeys a volume law 
but with a prefactor dependent on the temperature $T$.  For length 
scales below the correlation length, the von~Neumann entropy is 
dominated by the area law contribution.  Together, these combine 
to give entanglement entropy curves that have the qualitative behavior 
shown in Fig.~\ref{fig:Q_A-AdS}.  Note that in a local theory, lowering 
the temperature shows deviation from thermal behavior originating 
at small length scales.  Namely, the slope of $Q_A$ begins deviating 
from $0$ at small scales.  This entropy deviation at small scales is 
expected to be a general phenomenon of equilibrium states governed 
by a local Hamiltonian.

However, the entanglement entropy curves calculated for 
holographically FRW universes show drastically different behavior; 
see Fig.~\ref{fig:Q_w-FRW}.\footnote{It should be emphasized that we 
are calculating the entanglement entropy of the boundary state on the 
holographic screen, not the entropy associated with any bulk quantum 
fields.  We refer to the degrees of freedom on the screen that govern 
the background gravitation dynamics as the gravitational degrees of 
freedom.  Any low energy bulk excitations (which may include gravitons) 
are higher order corrections to the entanglement entropy and we do 
not discuss them.}
Namely, the deviations from maximal entropy originate at large length 
scales, and the entanglement entropy for small subregions is maximal 
regardless of the fluid parameter $w$.  Additionally, these entropy 
curves are invariant under time translation.  This behavior cannot 
be achieved by a local theory.  One may think that a Lifshitz type 
theory with large $z$ may be able to accommodate such behavior 
due to large momentum coupling, but the leading order contribution 
to the entanglement entropy in $d$ dimensions is believed to 
be proportional to $(L/\epsilon)^{d-1-1/z}$ for weakly coupled 
theories~\cite{MohammadiMozaffar:2017nri}, where $L$ is the characteristic 
length of the entangling region and $\epsilon$ is the cutoff length. 
Thus entanglement entropy is proportional to the volume only in the 
limit that $z \rightarrow \infty$, which would be a nonlocal field 
theory.  Indeed, entanglement entropy being maximal for small subregions 
is observed in a number of nonlocal theories~\cite{Barbon:2008ut,%
Fischler:2013gsa,Karczmarek:2013xxa,Shiba:2013jja,Fu:2016yrv,Liu:2017kfa} 
and is likely a generic phenomenon in such theories.

This leads us to believe that an appropriate holographic description 
of FRW universes would be nonlocal.%
\footnote{It is a logical possibility that a local theory could 
 exhibit volume law entropy behavior due to open dynamics.  Since the 
 size of the leaf is constantly growing, there are degrees of freedom 
 constantly being added to the system, which could already have long 
 range entanglement.  This seems to be an ad~hoc solution, and we will 
 not elaborate on this possibility further.}
This provides us with a few possibilities of theories that have the 
desired qualitative features, all of which have a freedom to tune 
a parameter which corresponds to changing $w$ (and hence the entropy):
\begin{enumerate}
\item[(I)]
a nonlocal theory with a characteristic length scale below the system 
size, changing the nonlocal length scale of the theory or energy of 
the state;
\item[(II)]
a nonlocal theory coupling sites together at all length scales (like 
a long-range interacting spin chain or a variant of the Sachdev-Ye-Kitaev 
model~\cite{Sachdev:1992fk,Kitaev:2015,Maldacena:2016hyu} with all-to-all 
random coupling between a fixed number, $q$, of sites, $\text{SYK}_{q}$), 
changing the energy of the state;
\item[(III)]
a nonlocal theory with a fundamental parameter controlling the coupling 
at all scales in which variations can change the entropy; for example, 
changing the number of sites coupled to each other in each term of the 
Hamiltonian (analogous to changing $q$ in $\text{SYK}_{q}$).
\end{enumerate}

The ground states of theories in case~(I) are explored in 
Refs.~\cite{Barbon:2008ut,Fischler:2013gsa,Karczmarek:2013xxa} in 
string theory frameworks.  This case can also be realized as a spin 
chain with interactions that couple all sites within a distance smaller 
than the characteristic nonlocal length scale.  Above the nonlocal 
length scale an area law term starts to pick up and will eventually 
dominate.  However, because of this eventual turn-on of an area law, 
the qualitative features of the entropy normalized by volume are 
different than those exhibited by FRW entropy curves.  Namely, the 
concavity of the $Q_A$ plot beyond the nonlocal length scale is opposite 
to that observed in the FRW case.  This is because beyond the nonlocal 
length scale the entropy approaches an area law, hence the second 
derivative of $Q_A$ will be positive, unlike that observed in the 
FRW case.  Raising the temperature will only add an overall constant 
asymptotic value to $Q_A$.  Hence, the concavity of $Q_A$ forbids 
the holographic theory of FRW spacetimes from being a theory with 
a characteristic nonlocal length scale smaller than the system size.

This reasoning leaves us with nonlocal theories with characteristic 
interaction lengths comparable to the system size---what does this 
mean?  It simply means that a site can be coupled to any other site. 
For simplicity we will consider SYK-like theories but rather than being 
zero dimensional we split up the degrees of freedom to live on a lattice 
but keep the random couplings between them.  At first thought, one may 
think that because of the random, all-to-all coupling the entanglement 
entropy for all subregions would always be maximal.  However this is not 
the case.  The entanglement entropy for small regions is indeed maximal, 
but then deviates at large length scales~\cite{Fu:2016yrv,Liu:2017kfa}. 
One can intuitively understand this by thinking about the $\text{SYK}_2$ 
model and Bell pairs.  The SYK couplings are random, and some sites 
will have significantly higher coupling than average.  In the ground 
state, these pairs have a high probability of being entangled, so if 
the subregion of interest contains only half of one of these special 
pairs, this will raise the entanglement with the outside.  However, 
once the subregion becomes larger there is a higher probability that 
a complete Bell pair is contained, and this will drop the entanglement 
entropy.

From this intuition, one can see that the ground state of SYK-like 
theories have near maximal entanglement for small regions, which then 
deviates at large length scales.  At higher energies, the probability 
of minimizing the term in the Hamiltonian coupling these special 
sites (and creating the effective Bell pair) will be lowered, and 
hence the entanglement entropy of all subregions will monotonically 
increase~\cite{Liu:2017kfa,Huang:2017nox}.  This behavior is reminiscent 
of that observed in FRW entanglement entropy if we relate the fluid 
parameter, $w$, to the energy of the nonlocal state:\ the case~(II) 
listed above.  The limit of $T \rightarrow \infty$ would then correspond 
to $w \rightarrow -1$.

The third possibility~(III) is similar to the one just discussed, 
but with the difference that $w$ is dual not to temperature but to 
a fundamental parameter dictating the ``connectivity'' of the boundary 
theory.  In the language of $\text{SYK}_q$, this would correspond to 
changing $q$, where $q$ is the number of coupled fermions in each 
interaction term of the Hamiltonian.  As $q$ increases, the ground 
state entanglement monotonically increases and as $q \rightarrow 
\infty$ becomes maximal.  This would be the limit corresponding to 
$w \rightarrow -1$.  However, any possibility like this, which employs 
a change of a fundamental parameter of the Hamiltonian, will require 
us to manufacture the whole Hilbert space of the boundary theory by 
considering the collection of only the low energy states for each value 
of $q$.  We would like one related class of spacetimes to be dual to 
one boundary theory, which is not the case in this option.  We thus 
focus on option~(II) as the best candidate, but we cannot logically 
exclude option~(III).

It is interesting to observe the relationship between where the deviation 
from volume law entropy occurs and where the corresponding spacetime 
emerges.  In the Schwarzschild-AdS case, $Q_A$ drops from $1$ immediately 
at small subregions, and the spacetime that emerges is precisely that 
which is reconstructed from small subregions.  Hence the directly 
reconstructable region appears at the boundary and grows inwards as 
the temperature of the state is lowered.  The converse is true in the 
case of FRW spacetimes.  As we move away from $w = -1$, the entanglement 
entropy drops from maximal at large subregions and the corresponding 
spacetime that emerges is constructed by intersecting large surfaces. 
This is because the HRRT surfaces of small subregions of leaves with 
$w$ near $-1$ all lie on the same codimension-1 surface, the future 
causal boundary of the leaf, analogous to the small surfaces in 
Fig.~\ref{fig:dS_2+1} in Appendix~\ref{subapp:HRRT-dS}.  The HRRT 
surfaces for large subregions deviate from this and hence allow for 
reconstructing a codimension-0 region starting with points deepest in 
the bulk.

The language of quantum error correction~\cite{Almheiri:2014lwa} and 
tensor networks~\cite{Swingle:2009bg,Pastawski:2015qua,Hayden:2016cfa} 
allows for a nice interpretation of this phenomenon.  The loss of 
entanglement in pure gravitational degrees of freedom affords nature 
the opportunity to redundantly encode local bulk degrees of freedom in 
the boundary.  In AdS, short range entanglement is lost first, and hence 
there is ``room'' for the information of local bulk degrees of freedom 
to be stored.  In the case of FRW, long range entanglement is lost first, 
and subsequently points in the bulk that require large subregions to 
reconstruct emerge first.

\section{Holographic Hilbert Spaces}
\label{sec:linearity}

The analysis of the previous sections brings us to a suitable position 
to discuss the structure of holographic Hilbert spaces.  In this section, 
we propose how a single theory can host states with different spacetime 
duals while keeping geometric operators linear in the space of 
microstates for a fixed semiclassical geometry.  We use intuition 
gathered from quantum thermodynamic arguments to guide us.  Similar 
ideas have been discussed in Ref.~\cite{Almheiri:2016blp}.  Here we 
present a slightly generalized argument to emphasize its independence 
of dynamics, and explain its application to our framework.

Let us assume that the entanglement entropy of subregions of a boundary 
state dual to a semiclassical geometry is calculated via the HRRT 
prescription.  Given a bulk spacetime, one can then find the corresponding 
entanglement entropies for all subregions of the boundary.  Note that 
here we consider the ``classical limit.''  Namely, all the subregions 
we consider contain $O({\cal N})$ degrees of freedom, where
\begin{equation}
  {\cal N} = \frac{{\cal A}}{4 l_{\rm P}^{d-1}},
\label{eq:cal-N}
\end{equation}
with ${\cal A}$ being the volume of the holographic space.  The collection 
of all boundary subregions and their corresponding entanglement entropies 
will be referred to as the entanglement structure of the state, which 
we denote by $S(\ket{\psi})$.

From here, it is natural to ask whether or not all states with the same 
entanglement structure are dual to the same bulk spacetime.  This might 
indeed be the case, but it leads to some undesirable features.  These 
primarily stem from the fact that given a particular entanglement 
structure, one can find a basis for the Hilbert space in which all 
basis states have the specified entanglement structure.  For a Hilbert 
space with a local product structure, one can do this by applying local 
unitaries to a state---these will retain the entanglement structure 
and yet generate orthogonal states.  This would imply that by generically 
superposing $e^{O({\cal N})}$ of these states, one could drastically 
alter the entanglement structure and create a state dual to a completely 
different spacetime.  Hence, geometric quantities could not be 
represented by linear operators, even in an approximate sense. 
If this were the case, a strong form of state dependence would be 
necessary to make sense of dynamics in the gravitational degrees 
of freedom~\cite{Nomura:2017npr}. 

However, it is not required that every state in the holographic Hilbert 
space with the same entanglement structure is dual to the same spacetime. 
How can this consistently happen?  Given an entanglement structure, 
$S(\ket{\phi})$, we expect the existence of a subspace in which generic 
states (within this subspace) have this same entanglement structure 
up to $O({\cal N}^p)$ corrections with $p < 1$.  The existence of 
a subspace with a unique entanglement structure is not surprising 
if the dimension of the subspace is $e^{O({\cal N}^p)}$ with ($p < 1$), 
since we generally expect
\begin{equation}
  S\biggl( \sum_{i=1}^{e^{M}} c_i \ket{\psi_i} \biggr) 
  = S(\ket{\psi}) + O(M),
\label{eq:EE_triv}
\end{equation}
where $S(\ket{\psi_i}) = S(\ket{\psi})$ for all $i$. 

However, we argue further that there exist such subspaces with 
dimension $e^{O({\cal N})}$, spanned by some basis states $\ket{\psi_i}$ 
($i = 1, \cdots, e^{Q {\cal N}}$), with
\begin{equation}
  S\biggl( \sum_{i=1}^{e^{Q {\cal N}}} c_i \ket{\psi_i} \biggr) 
  = S(\ket{\psi}) + O({\cal N}^p;\, p<1),
\label{eq:EE_nontriv}
\end{equation}
where $Q \leq 1$ does not scale with ${\cal N}$.  The existence of 
these subspaces with entanglement structures invariant under superpositions 
is expected from canonical typicality (also referred to as the general 
canonical principle)~\cite{Goldstein:2005aib,GCP}.  This provides us 
with the powerful result that generic states in subspaces have the same 
reduced density matrix for small subsystems (up to small corrections). 
The proof of this statement is purely kinematical and hence applies 
generally.  In fact, from canonical typicality the correction 
term in Eq.~(\ref{eq:EE_nontriv}) is exponentially small, 
$O(e^{-Q {\cal N}/2})$.

Canonical typicality is a highly nontrivial statement because the size 
of the subspaces in question is large enough that one would naively 
think that superpositions would ruin the entanglement structure at 
$O({\cal N})$.%
\footnote{Note that if one fine-tunes coefficients and selects states 
 in this subspace carefully, one could construct a state with lower 
 entanglement via superposition.}
Therefore, even if one considers an exponentially large superposition 
of microstates (so long as they are generic states from the same 
subspace), geometric operators can be effectively linear within this 
subspace.  We propose that states dual to semiclassical geometries 
are precisely generic states within their respective subspaces.

An example of one of these subspaces would be an energy band of an 
SYK theory.  These harbor an exponentially large number of states, 
and yet from canonical typicality any superposition of generic states 
within this band will have the same entanglement entropy.  Another 
example would be states that have energy scaling with the central 
charge, $c$, in AdS/CFT.  These are dual to large black holes and 
there are also an exponentially large number of states within the 
energy band.  Despite this, generic states within this energy band 
will have the same entanglement entropy structure.  Essentially, 
canonical typicality proves the existence of exponentially large 
subspaces that have entanglement structures preserved under 
superpositions of just as many states.

We need this strengthened statement because the entanglement entropy 
calculations for FRW suggest the size of subspaces dual to identical 
spacetimes are exponentially large.  This is because the quantity 
$Q$ in Eq.~(\ref{eq:EE_nontriv}) is related with von~Neumann 
entropies characterizing the whole state, e.g.\ $Q_A$ in 
Section~\ref{subsec:lower-T} with $A$ being the half boundary 
space and $Q_w(\pi/2)$ in Section~\ref{subsec:FRW-Q}.  This intuition 
stems from the statement that the thermal entropy density and 
entanglement entropy density for states in the thermodynamic limit 
are approximately equal.  For generic states within some energy 
interval subspace, this holds by canonical typicality.  The statement 
also results from assuming the system satisfies the ETH (like in 
AdS/CFT).  SYK models, however, do not strictly satisfy the ETH; 
nevertheless, it remains true that $Q_A$ at half system size gives 
a good approximation for the thermal entropy density, and the discrepancy 
vanishes as the energy of the states is increased.  For these reasons, 
we expect $Q_w(\pi/2)$ to well approximate the thermal entropy density 
of states dual to an FRW spacetime with fluid parameter $w$.

We can now address the properties of typical states within an entire 
Hilbert space.  Consider a holographic Hilbert space of a given theory, 
e.g.\ a CFT with a finite cutoff or the holographic theory of FRW 
spacetimes.  If there are multiple superselection sectors in a given 
theory, then we focus on one of them.  In such a Hilbert space, the 
effective subspace with $Q = 1$ corresponds to typical states.  Applying 
Page's analysis, we can then conclude that the only entanglement structure 
consistent with Eq.~(\ref{eq:EE_nontriv}) where $Q = 1$ must be that 
of maximal entropy.  For example, the number of microstates for a large 
black hole approaches the dimension of the boundary Hilbert space as 
$T \rightarrow \Lambda$, and these states are maximally entangled. 
Similarly, using the argument in the previous paragraph, the number 
of independent microstates in the de~Sitter limit approaches the 
dimension of the boundary Hilbert space, and these states are maximally 
entangled.  As shown in Section~\ref{sec:typical}, the directly 
reconstructable spacetime region vanishes in these cases---an effective 
subspace with $Q = 1$ does not have reconstructable spacetime.  It is 
in this sense that typical states in the {\it whole} Hilbert space have 
no reconstructable spacetime.

On the other hand, if $Q < 1$, the corresponding entanglement structure 
$S(\ket{\psi})$ can be non-maximal, and generic states in this 
subspace may be dual to some bulk spacetime.  As discussed in 
Section~\ref{subsec:locality}, we expect that dynamics of the boundary 
theory can naturally select these subspaces, for example by simply 
lowering the energy of the system in the case of the boundary CFT.

The structure discussed here allows for a single holographic Hilbert 
space to harbor effective subspaces dual to different geometries, 
allows for a ``generically linear'' spacetime operator, and hence 
eliminates the need for any strong form of state dependence.%
\footnote{By strong state dependence, we mean a theory that would 
 require state dependence to describe bulk excitations in the directly 
 reconstructable region of a boundary state which is a generic 
 superposition of states dual to a given spacetime.  For a more 
 detailed analysis of this statement, we refer the reader to 
 Ref.~\cite{Nomura:2017npr}.  The main result is that requiring 
 linearity for the multiple boundary representations of a bulk operator 
 is impossible if the number of geometry microstates is $e^{\cal N}$. 
 This prohibits the existence of a directly reconstructable region 
 for typical states.  Note that the directly reconstructable region 
 does not probe behind black hole horizons, and hence we are not 
 addressing the possibility that state dependence is necessary to 
 recover the black hole interior.}
Because this ``spacetime operator'' is identical for states of a given 
entanglement, it will obviously act linearly on generic superpositions 
of states within one of these dynamically selected, entanglement-invariant 
subspaces.  We suspect that it is only in this thermodynamic sense that 
classical spacetime emerges from the fundamental theory of quantum gravity.

\section{Conclusion}
\label{sec:conclusion}

\subsection{Discussion}
\label{subsec:discussion}

Our understanding of the relationship between spacetime and entanglement 
seems to be converging.  The necessity of entanglement between boundary 
degrees of freedom for the existence of spacetime has been known for some 
time, but this fact may have mistakenly established the intuition that 
the fabric of spacetime itself is purely this entanglement.  However, 
this cannot be the case.  A one-to-one mapping between the entanglement 
structure of a boundary state and the directly reconstructable bulk 
spacetime cannot be upheld in a state independent manner.  In addition, 
we see that as boundary entanglement approaches maximality the 
reconstructable region of the bulk vanishes.

In hindsight, this should not be too surprising.  Let us recall 
Van~Raamsdonk's discussion~\cite{VanRaamsdonk:2009ar} relating spacetime 
to entanglement by examining the link between mutual information and 
correlations in a system.  The mutual information between two boundary 
subsystems $A$ and $B$ is defined as
\begin{equation}
  I(A,B) = S(A) + S(B) - S(A \cup B).
\label{eq:mutual_info}
\end{equation}
This quantity bounds the correlations in a system between operators 
${\cal O}_A$ and ${\cal O}_B$, supported solely on $A$ and $B$ via 
the relation
\begin{equation}
  I(A,B) \ge \frac{(\langle {\cal O}_A {\cal O}_B \rangle 
    - \langle {\cal O}_A \rangle \langle {\cal O}_B \rangle)^2} 
    {2 |{\cal O}_A|^2 |{\cal O}_B|^2}.
\label{eq:mutual_info_bound}
\end{equation}
Hence, when the mutual information between two subregions $A$ and $B$ 
vanishes, the correlation between local operators supported within the 
subregions must also vanish.  Assuming that subregion duality holds, 
this implies that correlation functions of bulk fields vanish.  Generally, 
correlators between two bulk fields go as
\begin{equation}
  \langle {\cal O}_1(x_1) {\cal O}_2(x_2) \rangle \sim a f(L),
\label{eq:corr_geodesic}
\end{equation}
where $L$ is the distance of the shortest geodesic connecting $x_1$ 
and $x_2$, $a$ is some theory dependent constant, and $f(z)$ is 
a decreasing function of $z$.  One can then make the argument that 
decreasing entanglement between regions will drop the mutual information 
between the regions, and hence make $L$ effectively infinite.  This 
implies that the spacetime regions dual to subregions $A$ and $B$ 
are disconnected when the entanglement (and hence mutual information) 
vanishes.  For intuition's sake, one can imagine two subregions of the 
AdS boundary which are in a connected entanglement phase---increasing 
the distance between these two subregions will drop the mutual 
information.  This is an argument demonstrating the need for entanglement 
in a holographic theory dual to spacetime, so long as the holographic 
theory has subregion duality.

However, there is a different (quite the opposite) way to make the mutual 
information between small (less than half of the system) subregions 
vanish, and consequently kill the bulk correlations.  This is by 
considering maximally entropic boundary states---in these, the mutual 
information will vanish for any pair of subregions.  This is the case 
both in cutoff temperature AdS black holes and in the de~Sitter limit 
of the holographic theory of FRW universes.  In these, the boundary 
states are maximally entropic and hence the bulk correlators must 
vanish; however, there exist finite length geodesics in the bulk 
(even if restricted only to the directly reconstructable region) 
which connect all points on the boundary.  This means that the prefactor, 
$a$, of Eq.~(\ref{eq:corr_geodesic}) must vanish, making the bulk 
theory ultralocal.  In these cases, the maximal entropy implies that 
there cannot be an extra emergent bulk dimension.  This is because 
the ground state of any quantum field theory quantized on spacelike 
hypersurfaces must be entangled at arbitrarily short scales, which 
is violated by the assumption that $a = 0$.  However, this is not 
necessarily unexpected---in both de~Sitter space and cutoff temperature 
black holes, the directly reconstructable regions are codimension-1 
null surfaces of the bulk (the de~Sitter horizon and black hole 
horizon respectively).  A natural description of the fields on this 
surfaces would be through null quantization, which is known to 
be ultralocal~\cite{Wall:2011hj}.  Accordingly, we see a breakdown 
in the holographic description.

From the above arguments one can convince themselves that it is not 
entanglement itself which allows for the construction of spacetime, 
but rather something related to intermediate entanglement.

How can this be better understood?  The framework of tensor networks 
provides some intuition behind this.  Here, a maximally entropic 
boundary state is most naturally represented by a single bulk node 
with one bulk leg and multiple boundary legs.%
\footnote{Any attempt to create a bulk by artificially including more 
 nodes with extremely large bulk bond dimension can be reduced to the 
 case of one bulk node.}
Hence the ``spacetime'' is just one non-localizable bulk region, a 
``clump'' as defined in Ref.~\cite{Sanches:2017xhn}.  This bulk point 
can be reconstructed once a subregion of the boundary contains more 
than half of the boundary legs.  Here it is clear that a maximally 
entropic boundary state has no dual ``spacetime,'' and yet it is 
possible to encode a bulk code subspace with full recovery once more 
than half of the boundary is obtained.  Note that these typical states 
will all satisfy (in fact saturate) the holographic entropy cone 
inequalities~\cite{Bao:2015bfa} simply because a random tensor network 
accurately describes the state, but this does not mean that there is 
a reconstructable region of the spacetime.

Additionally, if maximally entropic states did have reconstructable 
spacetime, then state dependence would be necessary in order to describe 
bulk excitations in these states, under the assumption that subregion 
duality holds.  This is because the number of microstates with maximal 
entropy is approximately the dimension of the full boundary Hilbert 
space, and by the argument in Section~V.C. of Ref.~\cite{Nomura:2017npr}, 
it is impossible to find a boundary representation of a bulk operator 
that has support only on a subregion of the boundary {\it and} acts 
approximately linearly on all microstates of a given spacetime. 
Intuitively, this is because the operator will be over-constrained 
by insisting it both have support on a subregion of the boundary and 
act linearly on $D$ microstates, when the dimension of the full boundary 
space is $D$.  This means that if we require state independence, then 
the only possible boundary operators representing bulk excitations 
for a maximally entropic state must have support on the full boundary 
space.%
\footnote{This is not contradicting the statement in the previous 
 paragraph that the sole bulk node's state in a random tensor can be 
 recovered with just more than half of the boundary.  In that case, 
 only the recovery of the bulk code subspace for one microstate was 
 considered.  State independence would require us to have an operator 
 that acts linearly on {\it all} microstates of a given spacetime.}
Therefore, the minimum possible subregion in which bulk excitations 
can be encoded state independently is the whole boundary space; hence 
there is no directly reconstructable spacetime.  This directly highlights 
the tension between reconstructing spacetime for maximally entropic 
states (in any manner), and requiring both subregion duality and state 
independence.

But what happens if we lower the entanglement of the boundary state 
while keeping the dimension of the boundary Hilbert space constant? 
Again, we turn to tensor networks for intuition.  In these situations, 
a natural way to encode sub-maximal entanglement (while fixing the 
bulk leg dimension) is by including more bulk nodes.  Therefore, by 
reducing the boundary entanglement, it is possible to create a bulk 
code subspace in which subsystem recovery is possible.  It seems that 
quantum gravity naturally utilizes this sub-maximal entanglement in 
order to encode information via subregion duality.  This suggests that 
perhaps entanglement is not the fundamental constituent of spacetime 
per~se, but rather the avenue by which subregion duality manifests.

\subsection{Future directions}
\label{subsec:future-dir}

This paper has attempted to clarify the nature of spacetime in 
holographic theories and it naturally raises interesting questions 
to be investigated in future work.

\subsubsection*{Reconstructability and generalized holographic 
 renormalization}

The analysis of this paper utilized the condition for reconstructable 
spacetime presented in Ref.~\cite{Sanches:2017xhn}, but 
appropriately generalized for use in the context of holographic 
screens~\cite{Nomura:2017npr}.  This paper illuminated some highly 
desirable properties of the directly reconstructable region defined 
in this manner---namely that one can describe this region state 
independently.  It would be extremely beneficial to attempt to find 
an explicit way to construct bulk operators using this method, 
perhaps uniting it with the methods of entanglement wedge 
reconstruction~\cite{Faulkner:2017vdd,Cotler:2017erl}.

It would also be interesting to try and develop new tools for 
reconstructing the bulk.  The relationship between the depth in the 
bulk and the scale in the boundary theory in AdS/CFT suggests that 
it may be possible to define the reconstructable region of spacetime 
as that which is swept through a renormalization procedure.  How this 
manifests in general holography is not clear, but it is suggestive 
that there exists at least one foliation where one can ``pull'' the 
leaf inward while retaining the ability to consistently apply the 
HRRT prescription.  Because the area of these renormalized leaves 
are monotonically decreasing, it is natural that this ``pulling'' 
may correspond to some renormalization procedure. The decrease in 
area also happens locally, which can be seen by generalizing the 
spacelike monotonicity theorem of Ref.~\cite{Nomura:2016ikr}.

One guess as to how to construct the renormalized leaf is to first 
pick the coarse graining scale of the boundary, and then define the 
new leaf as the collection of all of the deepest points of the extremal 
surfaces anchored to subregions with the size of the coarse graining 
scale.  In AdS/CFT this will pull the boundary in along the $z$ 
direction as expected, while in FRW spacetimes this will pull the 
leaf along the null direction if the coarse graining scale is small. 
Using this method, one can renormalize to a given scale in a number 
of different ways.  For example, one could perform many small 
renormalization steps or one large one.  The renormalized leaves 
in the two cases will generically differ, and this may correspond 
to the difference between one-shot renormalization and a renormalization 
group method.  The collection of all renormalized leaves may then 
determine the reconstructable region.%
\footnote{Using this construction, it is not possible to extend 
 reconstruction beyond horizons, but it is possible to reach behind 
 entanglement shadows.}
Theorem~\ref{th:1} tells us that once the renormalized state becomes 
maximally entropic, the renormalization procedure must halt.  Furthermore, 
because extremal surfaces for non-maximally entropic states probe the 
bulk, this renormalization procedure will continue until the leaf becomes 
maximally entangled.  Thus, this renormalization group flow will halt 
only once a bifurcation surface or a null non-expanding surface is 
reached.  In this language, maximally entropic states correspond to 
fixed points.  This is speculation, but may shed some light on the 
nature of renormalization in general holographic theories.

\subsubsection*{Cosmic equilibration}

In Section~\ref{subsec:proof}, we proved that maximally entropic 
states have no directly reconstructable spacetime.  Additionally, 
we argued that if one desires a state on a holographic screen to be 
maximally entropic and evolve in time, then the holographic screen 
is a null non-expanding surface and the directly reconstructable 
region is no more than the screen itself.  This suggests that in 
a holographic theory of cosmological spacetimes, if a state becomes 
maximally entropic and the screen does not halt, then the holographic 
description approaches that of de~Sitter space.  Consequently, the 
area of the screen is constant.  It would be interesting to investigate 
the result from the other direction.  By first assuming that the 
screen approaches a constant area, one may be able to argue that 
the leaves would then approach maximal entropy, and hence the holographic 
description approaches that of de~Sitter space.  This could provide 
another way to consider equilibrating to de~Sitter type solutions; 
see Ref.~\cite{Carroll:2017kjo}.

\subsubsection*{Complementarity}

In Appendix~\ref{app:two-sided}, we highlighted the dependence of the 
reconstructable region on the frame of reference.  In the case of the 
two-sided AdS black hole, we considered different reference frames 
corresponding to different time slicings in the same boundary 
theory---as one shifts the difference in the two boundary times, 
one recovers more and more of the black hole interior.  This is 
an example of complementarity.  It would be interesting to pursue 
this idea further and investigate the directly reconstructable region 
of a two-sided black hole.

One intriguing aspect of the two-sided black hole is that the directly 
reconstructable region does not extend beyond the extremal surface 
barrier; this is a {\it macroscopic} distance away from the future 
singularity, regardless of the boundary frame.  Does this mean that 
the boundary CFT cannot describe semiclassical physics behind this 
barrier, even where curvature is small?  Perhaps this means that there 
is a different description for the interior, living on a different 
holographic space.

\subsubsection*{Fundamentality of subregion duality}

In many of the discussions throughout this paper, we either 
required subregion duality or saw that it naturally arose from 
other considerations.  This seems to suggest that subregion duality 
is a fundamental characteristic of general holography.  Investigating 
the manner in which subregion duality arises in AdS/CFT may shed 
light on holography in general spacetimes.

\subsubsection*{Holographic theory of flat FRW spacetimes}

One of the most obvious open problems is that of finding an effective 
holographic theory applicable beyond asymptotically AdS spacetimes. 
In this paper and throughout previous work, we have focused on the 
case of flat FRW universes and assumed that a theory exists on the 
holographic screen in which the generalized HRRT prescription holds. 
Investigations into this has led to a deeper understanding of the 
nature of entanglement in constructing spacetime, along with (the 
lack of) state dependence in holographic theories.

It seems that a consistent theory is possible, and the most promising 
candidate for a theory describing the gravitational degrees of freedom 
is a theory with long-range interactions in which the energy of the 
states are dual to the fluid parameter of the FRW universe.  We know 
that it cannot be entirely nonlocal because this would prohibit the 
existence of entanglement phase transitions.  A theory with long range 
interactions would accurately reproduce the entanglement entropy structure 
we observe for FRW universes and would allow for a universal theory 
describing the single class of spacetimes.  Beyond this, we have some 
additional data about the properties of the boundary theory.

We know that a code subspace of states manifests, and these states are 
dual to bulk excitations.  Assuming subregion duality holds, one can 
ask the question of whether or not nonlocality/very long-range interactions 
in the gravitational degrees of freedom prohibits the local propagation 
of bulk excitations in the boundary theory.  We expect that the operators 
dual to bulk excitations are weakly coupled to the gravitational degrees 
of freedom, and that a local description of these bulk operators exists 
in the boundary.  In fact, this is what happens when one renormalizes the 
AdS boundary down to a single AdS volume~\cite{Balasubramanian:1999jd}. 
This renormalization induces an infinite set of interactions which makes 
the resulting theory on the renormalized boundary nonlocal.  Despite 
this, the renormalized theory still describes bulk physics through 
subregion duality.  Hence, the nonlocality of the boundary theory 
does not seem to be a fundamental obstacle in describing low energy 
excitations using local dynamics in the boundary theory.%
\footnote{It would be interesting to study this effective boundary 
 theory, induced in AdS/CFT by renormalizing all the way down to the 
 AdS scale.  The holographic theory capturing sub-AdS locality could 
 be very closely related to the theory on holographic screens.}
The dynamics of boundary operators dual to bulk excitations in flat 
FRW spacetimes was studied in Ref.~\cite{Nomura:2017grg} and it was 
determined that regardless of dimension and fluid parameter, the 
spread of these operators was characteristic of a theory with $z=4$ 
Lifshitz scaling.  This provides extra constraints for finding 
a candidate theory.

\subsubsection*{Holographic theory for general spacetimes}

It might appear that defining quantum gravity using holography, as 
envisioned here, is background dependent.  Namely, the holographic 
theory is given for each class of background spacetimes, e.g.\ 
asymptotically AdS spacetimes and flat FRW spacetimes.  This situation 
is analogous to defining string theory on the worldsheet, which 
is defined separately on each target space background.  From the 
perspective of the worldsheet, different backgrounds correspond 
to different theories living on the two dimensional spacetime. 
Nevertheless, we believe there exists some unified framework 
encompassing all these possibilities.  Similarly, in the case 
of holographic theories, it is plausible that the resultant theories 
for different background spacetimes correspond to different sectors 
described within a single framework.

\section*{Acknowledgments}

We thank Chris Akers, Francisco Machado, Sanjay Moudgalya, and Fabio 
Sanches for useful discussions.  This work was supported in part 
by the National Science Foundation under grant PHY-1521446, by the 
Department of Energy, Office of Science, Office of High Energy Physics 
under contract No.\ DE-AC02-05CH11231, and by MEXT KAKENHI Grant 
Number 15H05895.

\appendix

\section{Reconstructability of Two-sided Black Holes and Complementarity}
\label{app:two-sided}

In the main part of the text, we have focused on spacetimes having a 
simply connected boundary.  It is interesting to consider when this 
is not the case and examine which (if any) results persist.  For 
definiteness, we here analyze the case of a two-sided eternal black 
hole in asymptotically AdS space.  In this case, the holographic screen 
is the union of the two asymptotic boundaries at spacelike infinity. 
The boundary theory comprises two CFTs, ${\rm CFT}_L$ and ${\rm CFT}_R$, 
which are decoupled from each other.  Hence, the Hamiltonian for the 
system is given by
\begin{equation}
  H_{\rm total} = H_L + H_R.
\end{equation}
The times $t_L$ and $t_R$ associated respectively with $H_L$ and $H_R$ 
run in opposite directions along the two asymptotic boundaries.

Since the theories are decoupled, it might appear that one could 
evolve each of the theories independently---effectively foliating 
the holographic screen by two independent parameters, $(t_L, t_R)$. 
Per the construction outlined in Section~\ref{subsec:reconst}, 
the directly reconstructable region would then be the union of all 
points localized by intersecting entanglement wedges of HRRT surfaces 
individually anchored to ``one'' leaf, each of which is labeled by 
$(t_L, t_R)$.  Here, ``one'' leaf corresponds to picking a connected, 
equal time slice of the left boundary and independently a connected, 
equal time slice of the right boundary.  If this were the case, the 
reconstructable region would be most of the spacetime, including a 
macroscopic portion of the interior (aside from a region near the 
singularity with $r < r_+/2^{1/d}$, where $r_+$ is the horizon 
radius)~\cite{Susskind:2014moa}.

However, a theory described by Hamiltonian dynamics should have a single 
time parameter.  To make the holographic theory compatible with this, 
we postulate that there is a single parameter $t$ that foliates the 
multiple disconnected components of the holographic screen.  From this 
assumption, there are multiple suitable foliations, and among them 
we must pick one---this corresponds to choosing a reference frame, 
a gauge for the holographic redundancy~\cite{Nomura:2011rb}.  In the 
case of a two-sided black hole, this gives us a one parameter family 
of foliations corresponding to the freedom in choosing the relative 
time shift between $t_L$ and $t_R$ in the CFTs, even after choosing 
a natural foliation at each boundary.

In general, each of these individual foliations reconstruct a different 
region of the bulk spacetime.  For example, adopting the usual thermofield 
double state construction~\cite{Maldacena:2001kr} corresponds to choosing 
a reference frame
\begin{equation}
  t_L = t_R = t,
\label{eq:thermo}
\end{equation}
in which the $t = 0$ slice in the bulk is the one passing through the 
bifurcation surface.  Since time translation is a Killing symmetry in 
this spacetime, and the bifurcation surface is invariant under this 
translation, the HRRT surfaces for any time $t$ never enter the interior 
of the black hole.  Connected HRRT surfaces always pass through the 
bifurcation surface in such a situation (unless the subregion has support 
on only one of the boundaries, in which case the HRRT surface stays in 
one side of the black hole).  The reconstructable region in this reference 
frame, therefore, does {\it not} include the interior of the black hole.

However, one could alternatively consider a reference frame in which 
there is a relative shift in the two times
\begin{equation}
  t_L = t + \Delta,
\quad
  t_R = t.
\label{eq:non-thermo}
\end{equation}
In this case, the connected HRRT surfaces would not necessarily pass 
through the bifurcation surface and could probe regions of the interior, 
and hence parts of the interior will be reconstructable.  We can interpret 
this foliation dependence of the reconstructable region as a version of 
complementarity~\cite{Susskind:1993if}.  In this light, the canonical 
thermofield double time foliation corresponds to an entirely exterior 
description of the black hole, while increasing $\Delta$ allows for 
more of the region behind the horizon to be reconstructed.  An important 
point is that we should not consider leaves with different $\Delta$'s 
in a single description---they correspond to different descriptions 
in different reference frames.  We also note that regardless of the 
foliation, we cannot reconstruct near the singularity because of the 
extremal surface barrier located at $r = r_+/2^{1/d}$.  This suggests 
that in order to probe physics of the singularity we must use a 
different method.

With this interpretation of bulk reconstruction, we would like to 
examine whether or not spacetime ``disappears'' as we approach maximal 
entropy.  A priori, it seems that a macroscopic spacetime region would 
remain as we increase the black hole radius because some portion of 
the interior is reconstructable.  However, this apparent contradiction 
is resolved by considering a finite coordinate time interval and 
examining the reconstructable volume as one increases the temperature.

Consider any foliation where the relative time shift between $t_L$ and 
$t_R$ has been fixed.  In order to carry out the analysis analogous 
to Section~\ref{subsec:AdS-BH}, we fix an interval of coordinate time 
$\varDelta t$ and fix the cutoff surface at $r = R$.  Increasing the 
temperature of the black hole moves the horizon closer and closer to 
the cutoff surface, which can be represented in the Penrose diagram 
as in Fig.~\ref{fig:2-sided}.
\begin{figure}[t]
\begin{center}
  \includegraphics[height=6.5cm]{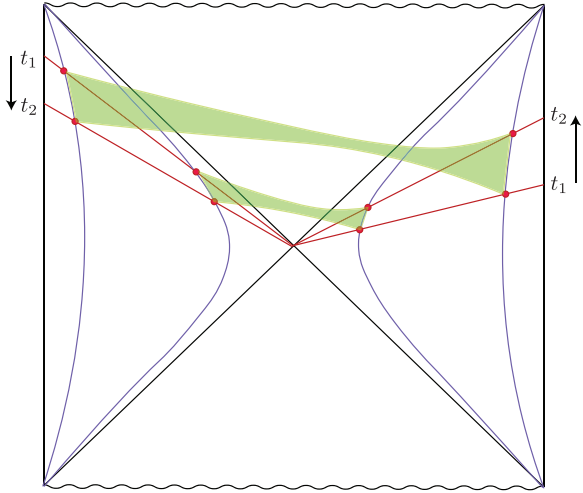}
\end{center}
\caption{The spacetime regions reconstructable using connected HRRT 
 surfaces anchored to subregions with support on both asymptotic 
 boundaries within the range $t \in [t_1, t_2]$ are depicted (green 
 shaded regions) for two different values of black hole horizon radius 
 $r_+$ in a two-sided eternal AdS black hole.  The holographic screen 
 (blue) in both cases is the cutoff surface $r = R$.  Here, we superimpose 
 the respective Penrose diagrams in the two cases to compare the amount 
 of reconstructable spacetime volume available by allowing connected 
 HRRT surfaces.}
\label{fig:2-sided}
\end{figure}
The allowed range of times is depicted by the constant time surfaces 
$t_1$ and $t_2$.  As we take the limit $r_+ \rightarrow R$, which 
corresponds to taking the temperature of the black hole $T_{\rm H} 
\rightarrow \Lambda$ where $\Lambda$ is the cutoff in the boundary 
theory, the finite range of time collapses down to the bifurcation 
surface on both sides.  Thus, the relative reconstructable spacetime 
volume shrinks to zero.

We find that our claim persists despite the addition of a disconnected 
boundary region that allows for the reconstruction of spacetime behind 
a black hole horizon.

\section{Calculations for the Schwarzschild-AdS Spacetime}
\label{app:S-AdS}

In this appendix, we provide explicit calculations of the spatial 
volume and HRRT surfaces of the Schwarzschild-AdS spacetime.

\subsection{Reconstructable volume}
\label{subapp:S-AdS_volume}

The Schwarzschild-AdS spacetime in $d+1$ dimensions is described by 
the metric
\begin{equation}
  ds^2 = - \biggl( \frac{r^2}{l^2} + 1 - \frac{2\mu}{r^{d-2}} \biggr) dt^2 
    + \frac{dr^2}{\frac{r^2}{l^2} + 1 - \frac{2\mu}{r^{d-2}}} 
    + r^2 d\Omega_{d-1}^2,
\label{eq:AdS-BH}
\end{equation}
where $l$ is the AdS radius, and $\mu$ is related with the black hole 
horizon radius $r_+$ as
\begin{equation}
  2\mu = \frac{r_+^d}{l^2} \biggl( 1 + \frac{l^2}{r_+^2} \biggr).
\label{eq:mu}
\end{equation}
The Hawking temperature of the black hole is given by
\begin{equation}
  T_{\rm H} = \frac{d r_+^2 + (d-2)l^2}{4\pi r_+ l^2}.
\label{eq:T_H}
\end{equation}

Consider a large AdS black hole $r_+ \gg l$.  In this limit,
\begin{equation}
  2\mu = \frac{r_+^d}{l^2},
\qquad
  T_{\rm H} = \frac{d r_+}{4\pi l^2},
\label{eq:large-BH}
\end{equation}
and the metric is well approximated by
\begin{equation}
  ds^2 = - \biggl( \frac{r^2}{l^2} - \frac{r_+^d}{l^2 r^{d-2}} \biggr) dt^2 
    + \frac{dr^2}{\frac{r^2}{l^2} - \frac{r_+^d}{l^2 r^{d-2}}} 
    + r^2 d\Omega_{d-1}^2.
\label{eq:metric-large}
\end{equation}
Let us now introduce an infrared cutoff $r \leq R$ and consider the 
spatial volume between the black hole horizon and the cutoff
\begin{align}
  V(r_+,R) &= A_{d-1} \int_{r_+}^R\! 
    \frac{r^{d-1}}{\sqrt{\frac{r^2}{l^2} - \frac{r_+^d}{l^2 r^{d-2}}}}\, dr 
\nonumber\\
  &= \frac{2\pi^{d/2}}{\Gamma(d/2)}\, l\, r_+^{d-1}\! 
    \int_1^{\frac{R}{r_+}}\! \frac{x^{d-2}}{\sqrt{1-\frac{1}{x^d}}}\, dx,
\label{eq:volume}
\end{align}
where $A_{d-1} = 2\pi^{d/2}/\Gamma(d/2)$ is the area of the 
$(d-1)$-dimensional unit sphere.  Here, we have focused on the 
spatial volume because the system is static.

We normalize this volume by the volume of the region $r \leq R$ in 
empty AdS space
\begin{align}
  V(R) &= A_{d-1} \int_0^R\! 
    \frac{r^{d-1}}{\sqrt{\frac{r^2}{l^2} + 1}}\, dr 
\nonumber\\
  &= \frac{2\pi^{d/2}}{(d-1)\Gamma(d/2)}\, l\, R^{d-1},
\label{eq:volume-empty}
\end{align}
where we have used $R \gg l$ in the second line.  This gives us the 
quantity quoted in Eq.~(\ref{eq:vol-ratio}):
\begin{equation}
  f\Bigl(\frac{r_+}{R}\Bigr) \equiv \frac{V(r_+,R)}{V(R)} 
  = (d-1) \frac{r_+^{d-1}}{R^{d-1}}\! 
    \int_1^{\frac{R}{r_+}}\! \frac{x^{d-2}}{\sqrt{1-\frac{1}{x^d}}}\, dx.
\label{eq:vol-ratio-app}
\end{equation}

\subsection{HRRT surfaces}
\label{subapp:HRRT_S-AdS}

Consider a large black hole in asymptotically AdS space.  The holographic 
theory is then a CFT.  Suppose the temperature of the system $T$ is 
lower than the cutoff scale, $T < \Lambda$.  Here we study the behavior 
of the von~Neumann entropy of a spherical cap region $A$ on $r = R$ 
in this setup.

The region is specified by a half opening angle $\psi$
\begin{equation}
  0 \leq \theta \leq \psi,
\label{eq:psi-def}
\end{equation}
where $\theta$ is a polar angle parameterizing $S^{d-1}$ with constant 
$t$ and $r$.  The HRRT surface $\gamma_A$ is then given by function 
$r(\theta)$, which is determined by minimizing the area functional:
\begin{equation}
  \norm{\gamma_A} = \underset{r(\theta)}{\rm min} \left[ A_{d-2} 
    \int_0^\psi\! r^{d-2}\, \sin^{d-2}\!\theta\, 
    \sqrt{ r^2 + \frac{(\frac{dr}{d\theta})^2}
      {\frac{r^2}{l^2} + 1 - \frac{2\mu}{r^{d-2}}} }\, d\theta \right],
\label{eq:norm-gamma_A}
\end{equation}
with the boundary condition
\begin{equation}
  r(\psi) = R,
\label{eq:bc-r}
\end{equation}
where $A_{d-2}$ is the area of the $(d-2)$-dimensional unit sphere, 
and $\mu$ is given by Eq.~(\ref{eq:AdS-BH}).  Here and below, we assume 
$\psi \leq \pi/2$.  For $\psi > \pi/2$, the entropy of $A$ is determined 
by $S(\psi) = S(\pi-\psi)$.

The surface $\gamma_A$ is well approximated to consist of two components:\ 
(i) a ``cylindrical'' piece with $\theta = \psi$, which is perpendicular 
to the cutoff surface $r = R$ and extends down to $r = r_0$ ($< R)$ and 
(ii) the ``bottom lid'' with $r = r_0$ and $0 \leq \theta \leq \psi$; 
see Fig.~\ref{fig:S-AdS_HRRT}.
\begin{figure}[t]
\begin{center}
  \includegraphics[height=6.5cm]{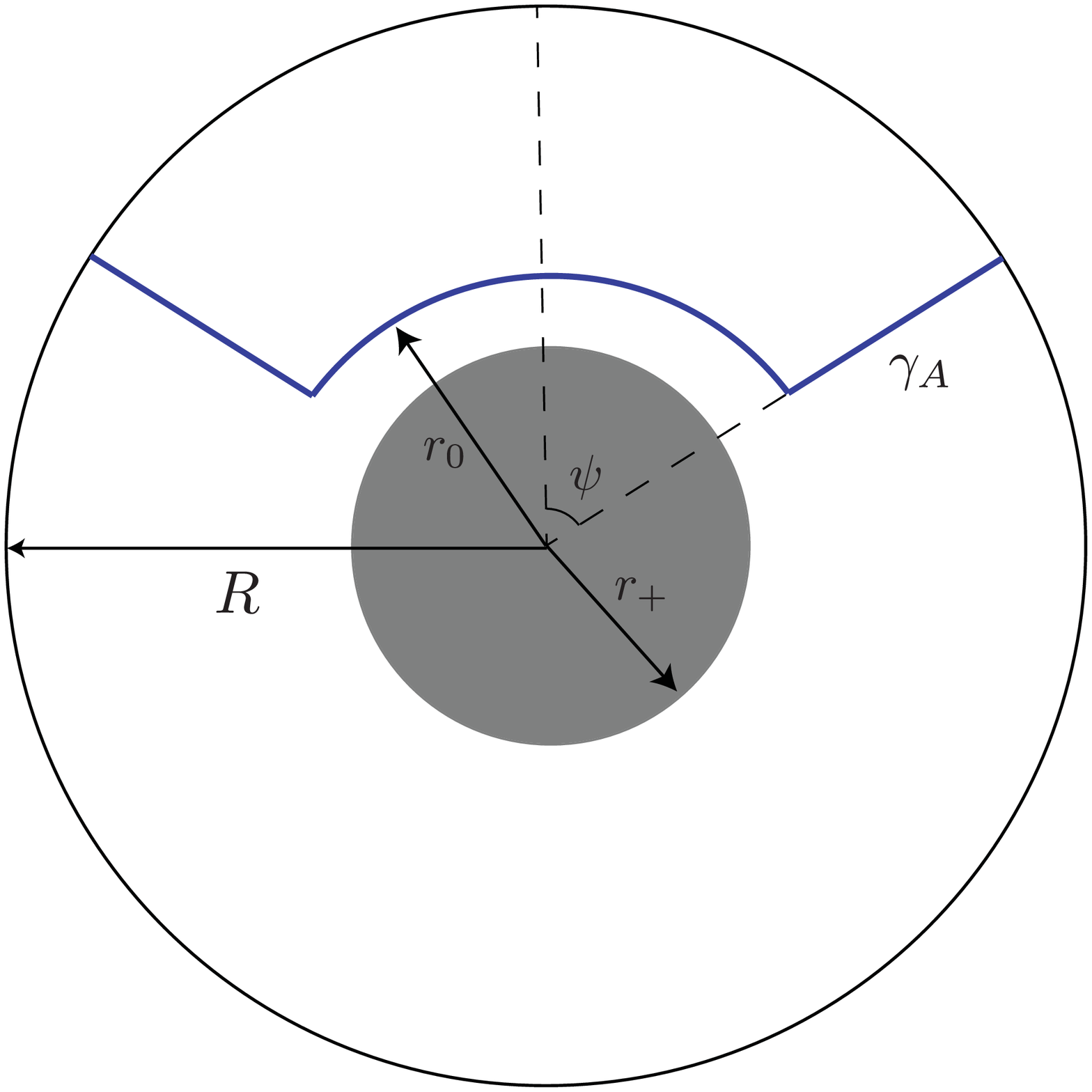}
\end{center}
\caption{The HRRT surface $\gamma_A$ in the Schwarzschild-AdS 
 spacetime can be well approximated by consisting of two components:\ 
 a ``cylindrical'' piece with $\theta = \psi$ and a ``bottom lid'' 
 piece with $r = r_0$.}
\label{fig:S-AdS_HRRT}
\end{figure}
The area of the surface is then given by
\begin{equation}
  \norm{\gamma_A} = \underset{r_0}{\rm min} \left[ A_{d-2}\, 
    \sin^{d-2}\!\psi \int_{r_0}^R\! \frac{r^{d-2}}{\sqrt{\frac{r^2}{l^2} 
      - \frac{r_+^d}{l^2 r^{d-2}}}}\, dr 
    + A_{d-2}\, r_0^{d-1} \int_0^\psi\! \sin^{d-2}\!\theta\, 
      d\theta \right],
\label{eq:norm-g_A-approx}
\end{equation}
where $r_+$ is the horizon radius, and we have used the approximation 
that $r_+ \gg l$ and hence Eq.~(\ref{eq:large-BH}).  The value of $r_0$ 
is determined by the minimization condition
\begin{equation}
  \sqrt{r_0^2 - \frac{r_+^d}{r_0^{d-2}}} 
  = \frac{\sin^{d-2}\!\psi}{(d-1) \int_0^\psi\! \sin^{d-2}\!\theta\, 
    d\theta} l.
\label{eq:r_0}
\end{equation}

As discussed in Section~\ref{subsec:AdS-BH}, the cutoff at $r = R$ in 
our context simply means that the renormalization scale in the boundary 
theory is lowered; in particular, it does not mean that the theory is 
modified by actually terminating space there.  The length in the boundary 
theory, therefore, is still measured in terms of the $d$-dimensional 
metric at infinity, $r = \infty$, with the conformal factor stripped 
off.  The radius of the region $A$ is then given by
\begin{equation}
  L = l\, \psi,
\label{eq:L-psi}
\end{equation}
and not $R \psi$.  Since the cutoff length is $1/\Lambda \approx l^2/R$, 
we should only consider the region $\psi \gtrsim l/R$.

The solution of Eq.~(\ref{eq:r_0}) behaves as
\begin{alignat}{4}
  \mbox{(i)} \quad &
  r_0 = \frac{l}{\psi}\,\, (\gg r_+) \quad &&
  \mbox{for } \frac{l}{R} < \psi \ll \frac{l}{r_+},
\label{eq:sol-1}\\
  \mbox{(ii)} \quad &
  r_0 - r_+ = \frac{l^2}{d \psi^2 r_+}\,\, (\ll \frac{r_+}{d}) \quad &&
  \mbox{for } \frac{l}{r_+} \ll \psi \ll 1,
\label{eq:sol-2}\\
  \mbox{(iii)} \quad &
  r_0 - r_+ = O(1)\, \frac{l^2}{r_+} &&
  \mbox{for } \psi \approx O(1).
\label{eq:sol-3}
\end{alignat}
In the case of (i), $\norm{\gamma_A}$ is dominated by the first term in 
Eq.~(\ref{eq:norm-g_A-approx}), so that
\begin{equation}
  \norm{\gamma_A} = \frac{A_{d-2}}{d-2} l R^{d-2} \psi^{d-2}.
\label{g_A-case-1}
\end{equation}
Here and below, we assume $d > 2$.  We thus obtain an area law for 
the entropy
\begin{equation}
  S_A = \frac{\norm{\gamma_A}}{4 l_{\rm P}^{d-1}} 
  \approx c A_{d-2} L^{d-2} \Lambda^{d-2},
\label{eq:S_A-case-1}
\end{equation}
where $c \approx (l/l_{\rm P})^{d-1}$ is the central charge of the 
boundary CFT.

In the case of (ii), $\norm{\gamma_A}$ is given by
\begin{equation}
  \norm{\gamma_A} = \frac{A_{d-2}}{d-2} l R^{d-2} \psi^{d-2} 
  + \frac{A_{d-2}}{d-1} r_+^{d-1} \psi^{d-1}.
\label{g_A-case-2}
\end{equation}
We find that the first (second) term is larger for
\begin{equation}
  \psi < (>)\, \frac{d-1}{d-2} \frac{l R^{d-2}}{r_+^{d-1}},
\label{eq:psi-trans}
\end{equation}
so that the entanglement entropy behaves as
\begin{equation}
  S_A \approx 
  \left\{ \begin{array}{ll}
    c A_{d-2} L^{d-2} \Lambda^{d-2} &
    \mbox{for } L \ll L_*, \\
    c A_{d-2} \frac{r_+^{d-1} L^{d-1}}{l^{2d-2}} \approx 
      c \left(\frac{T}{\Lambda}\right)^{d-1} A_{d-2}L^{d-1}\Lambda^{d-1} &
    \mbox{for } L \gg L_*,
  \end{array} \right.
\label{eq:S_A-case-2}
\end{equation}
where
\begin{equation}
  L_* \approx \frac{l^2 R^{d-2}}{r_+^{d-1}} 
  \approx \frac{\Lambda^{d-2}}{T^{d-1}}.
\label{eq:L_tr-app}
\end{equation}

For $\psi \approx O(1)$, i.e.\ case~(iii), we find
\begin{equation}
  S_A \approx c \left(\frac{T}{\Lambda}\right)^{d-1} 
    A_{d-2} L^{d-1} \Lambda^{d-1}.
\label{eq:S_A-case-3}
\end{equation}
Combining the results in all three cases gives the expression in 
Eqs.~(\ref{eq:S_A-case},~\ref{eq:L_tr}).

\section{Calculations for the de~Sitter Limit of FRW Universes}
\label{app:dS-FRW}

This appendix collects explicit calculations for entropies and HRRT 
surfaces in the de~Sitter limit of FRW spacetimes.

\subsection{Entropies in the case of {\boldmath $(2+1)$}-dimensional bulk}
\label{subapp:2Dproof}

Here we see that for $(2+1)$-dimensional FRW spacetimes, the results 
of Ref.~\cite{Nomura:2016ikr} immediately tell us that the entanglement 
entropy of an arbitrary (not necessarily connected) subregion $A$ is 
maximal in the de~Sitter limit:
\begin{equation}
  S_{A,w \rightarrow -1} = \frac{1}{4l_{\rm P}} 
    {\rm min}\{ \norm{A}, \norm{\bar{A}} \}.
\label{eq:2+1_sat}
\end{equation}

Consider an FRW universe in $d+1$ dimensions dominated by a single 
ideal fluid component with the equation of state parameter $w = p/\rho$ 
($|w| \leq 1$).  From the analysis of Ref.~\cite{Nomura:2016ikr}, we 
know that the holographic entanglement entropy of a spherical cap region 
$A$ on a leaf---parameterized by the half opening angle $\psi$ as viewed 
from the center of the bulk---scales with the smaller of the volumes 
of $A$ and $\bar{A}$.  The proportionality constant
\begin{equation}
  Q_w(\psi) \equiv \frac{S(\psi)}{\frac{1}{4 l_{\rm P}^{d-1}} 
    {\rm min}\{ \norm{A}, \norm{\bar{A}} \}},
\label{eq:def-Q}
\end{equation}
satisfies the properties
\begin{equation}
  Q_w(\psi \rightarrow 0) \rightarrow 1,
\qquad
  Q_{w \rightarrow -1}(\psi) \rightarrow 1,
\label{eq:prop-Q-1}
\end{equation}
\begin{equation}
  \frac{\partial Q_w(\psi)}{\partial \psi} 
    \biggr|_{\psi = 0} = 0,
\qquad
  \frac{\partial Q_w(\psi)}{\partial \psi} 
    \biggr|_{\psi < \frac{\pi}{2}} \leq 0,
\qquad
  \frac{\partial Q_w(\psi)}{\partial w} < 0.
\label{eq:prop-Q-2}
\end{equation}
(The original analysis was performed for $(3+1)$-dimensional FRW universes, 
but these properties persist in arbitrary spacetime dimensions.)

The second relation in Eq.~(\ref{eq:prop-Q-1}) implies that in the 
de~Sitter limit, $w \rightarrow -1$, the holographic entanglement entropy 
of a spherical cap region is maximal.  Now, consider $(2+1)$-dimensional 
FRW universes, in which a leaf has only one spatial dimension.  We 
consider a subregion on the leaf consisting of the union of two small 
intervals $A$ and $B$.  Note that a similar setup is often discussed 
in AdS/CFT, where two possible extremal surfaces homologous to the 
subregion compete, so that a phase transition from the disconnected 
to connected HRRT surfaces occurs as the regions $A$ and $B$ are taken 
to be closer; see Fig.~\ref{fig:phase-tr}.  We want to understand what 
happens in the case of FRW spacetimes.
\begin{figure}[t]
\begin{center}
  \setcounter{subfigure}{0}
  \subfigure[]
     {\includegraphics[width=6cm]{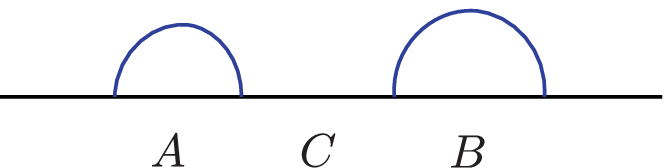}}
\hspace{1.5cm}
  \setcounter{subfigure}{1}
  \subfigure[]
     {\includegraphics[width=6cm]{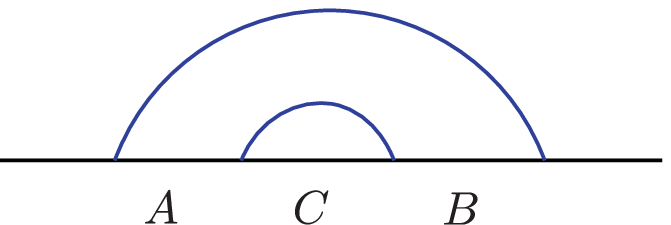}}
\end{center}
\caption{Two possible extremal surfaces anchored to the boundary of 
 a subregion $AB$ on a leaf, given by the union of two disjoint intervals 
 $A$ and $B$.  The areas of the surfaces depicted in (a) and (b) are 
 denoted by $E_{\rm disconnected}(AB)$ and $E_{\rm connected}(AB)$, 
 respectively.}
\label{fig:phase-tr}
\end{figure}

We denote the areas of two possible extremal surfaces by
\begin{align}
  E_{\rm disconnected}(AB) &= E(A) + E(B)
\nonumber\\
  &= Q_w(A)\, \norm{A} + Q_w(B)\, \norm{B},
\label{eq:E_disconn}
\end{align}
and
\begin{align}
  E_{\rm connected}(AB) &= E(ABC) + E(C)
\nonumber\\
  &= Q_w(ABC)\, \norm{ABC} + Q_w(C)\, \norm{C},
\label{eq:E_conn}
\end{align}
where $A$, $B$, and $C$ are defined in Fig.~\ref{fig:phase-tr}.  A phase 
transition can occur when
\begin{equation}
  E_{\rm disconnected}(AB) = E_{\rm connected}(AB).
\label{eq:phase-tr}
\end{equation}

The condition of Eq.~(\ref{eq:phase-tr}) can be satisfied for any 
$w$ away from the de~Sitter limit because of the second relation in 
Eq.~(\ref{eq:prop-Q-2}).  Since a larger region has a greater volume 
but also has a smaller coefficient, it is possible for the two extremal 
surfaces to compete.  However, in the de~Sitter limit the requirement 
for a phase transition becomes
\begin{equation}
  \norm{ABC} + \norm{C} = \norm{A} + \norm{B},
\label{eq:dS-compete}
\end{equation}
which is clearly impossible because the left hand side is always 
greater.  Since a general subregion of the leaf is a union of 
disconnected intervals, the above argument implies that the entanglement 
entropy is merely the sum of each interval's volume for sufficiently 
small regions.  Extending the argument to large regions in which their 
complements matter, we can conclude that arbitrary subregions have 
maximal entanglement entropies in a $(2+1)$-dimensional de~Sitter 
universe.

\subsection{Entropies in the {\boldmath $w \rightarrow -1$} limit of 
 FRW spacetimes}
\label{subapp:FRWlimit}

The global spacetime structure in the case of a single fluid component 
with $w \neq -1$ is qualitatively different from the case discussed mainly 
in Section~\ref{subsec:dS_max-ent}, i.e.\ the case in which a universe 
approaches de~Sitter space at late times.  Nevertheless, here we show 
that the holographic entanglement entropy of an arbitrary subregion on 
a leaf becomes maximal in the $w \rightarrow -1$ limit.

Let us consider an FRW universe filled with a single fluid component 
with the equation of state $w$.  The scale factor is then given by
\begin{equation}
  a(t) = c\, t^{\frac{2}{d(1+w)}},
\label{eq:a-t}
\end{equation}
where $c > 0$ is a constant.  We focus on a leaf $\sigma_*$ at time 
$t_*$ and the causal region $D_{\sigma_*}$ associated with it.  Following 
Ref.~\cite{Nomura:2016ikr}, we perform $t_*$-dependent coordinate 
transformation on the FRW time and radial coordinates $t$ and $r$:
\begin{align}
  \eta &= \frac{2}{d-2+dw} \left\{ 
    \left( \frac{t}{t_*} \right)^{\frac{d-2+dw}{d(1+w)}} - 1 \right\},
\label{eq:eta}\\
  \rho &= \frac{2}{d(1+w)} c\, t_*^{-\frac{d-2+dw}{d(1+w)}} r.
\label{eq:rho}
\end{align}
This converts the metric into the form
\begin{equation}
  ds^2 = \biggl( \frac{{\cal A}_*}{A_{d-1}} \biggr)^{\frac{2}{d-1}} 
    \biggl( \frac{d-2+dw}{2}\eta + 1 \biggr)^{\frac{4}{d-2+dw}} 
    \bigl( -d\eta^2 + d\rho^2 + \rho^2 d\Omega_{d-1}^2 \bigr),
\label{eq:met-transf}
\end{equation}
where $A_{d-1}$ is the area of the $(d-1)$-dimensional unit sphere, 
defined below Eq.~(\ref{eq:volume}), and ${\cal A}_*$ is the volume 
of the leaf $\sigma_*$
\begin{equation}
  {\cal A}_* = \biggl( \frac{d(1+w)}{2} \biggr)^{d-1} 
    \! A_{d-1} \, t_*^{d-1}.
\label{eq:A_*}
\end{equation}
In these coordinates, $D_{\sigma_*}$ is mapped into the region $\eta \in 
[-1,1]$ and $\rho \in [0,1-|\eta|]$.%
\footnote{For $w \geq -1 + 4/d$, the region $D_{\sigma_*}$ hits the 
 big bang singularity, so we need to restrict our attention to a 
 portion of $D_{\sigma_*}$, e.g.\ $D^+_{\sigma_*} = \{p \in D_{\sigma_*} 
 \:|\: t(p) \geq t_* \}$.  This issue is not relevant to our discussion 
 here.}

We can now take $w = -1+ \epsilon$ in Eq.~(\ref{eq:met-transf}) and 
expand it around $\epsilon = 0$.  This gives
\begin{equation}
  ds^2 = \biggl( \frac{{\cal A}_*}{A_{d-1}} \biggr)^{\frac{2}{d-1}} 
    \biggl( \frac{1}{(1-\eta)^2} 
      - d \frac{\eta+(1-\eta)\ln(1-\eta)}{(1-\eta)^3} \epsilon 
      + \cdots \biggr) 
    \bigl( -d\eta^2 + d\rho^2 + \rho^2 d\Omega_{d-1}^2 \bigr).
\label{eq:met-expand}
\end{equation}
The leading order term describes the causal region inside a leaf of 
volume ${\cal A}_*$ in de~Sitter space with conformal coordinates.  The 
time translational Killing symmetry in these coordinates is
\begin{align}
  \eta &\rightarrow a\eta + 1 - a, \\
  \rho &\rightarrow a \rho.
\label{eq:dS-iso}
\end{align}
The expansion in Eq.~(\ref{eq:met-expand}) is not valid when $\eta 
\lesssim 1-\epsilon$.  However, this occurs only for a small subset 
of all the subregions on $\sigma_*$, which becomes measure zero when 
$\epsilon \rightarrow 0$.  Continuity then tells us that the entanglement 
entropy $S_A$ of any subregion $A$ on $\sigma_*$ takes the same value 
as that calculated in de~Sitter space in the $\epsilon \rightarrow 0$ 
limit.  However, we have already concluded from the argument in 
Section~\ref{subsec:dS_max-ent} that the entanglement entropies take 
the maximal form in de~Sitter space, hence
\begin{equation}
  S_A \underset{w \rightarrow -1}{\longrightarrow} 
    \frac{1}{4 l_{\rm P}^{d-1}} {\rm min}\{ \norm{A}, \norm{\bar{A}} \}.
\label{eq:S_A-norm_A-limit}
\end{equation}
Note that the area of the leaf, ${\cal A}_*$, keeps growing indefinitely, 
so that $D_{\sigma_*}$ at each time $t_*$ is mapped to a different 
auxiliary de~Sitter space.  The ratio $Q_w(A) = S_A/({\rm min}\{ \norm{A}, 
\norm{\bar{A}} \} / 4 l_{\rm P}^{d-1})$, however, depends only on $w$ 
and not $t_*$.

\subsection{HRRT surfaces}
\label{subapp:HRRT-dS}

Here we present two examples in which one can analytically see the 
convergence of the HRRT surfaces onto the future boundary of the causal 
region of a leaf in the de~Sitter limit.

\subsubsection*{The de~Sitter limit of FRW universes in {\boldmath $2+1$} 
 dimensions}

As the first example, consider the de~Sitter limit of FRW universes 
in $2+1$ dimensions
\begin{equation}
  ds^2 = a^2(\eta)\, (-d\eta^2 + dx^2 + dy^2).
\label{eq:dS_2+1}
\end{equation}
Here, $\eta \in (-\infty,0)$ is the conformal time, and the scale factor 
is given by
\begin{equation}
  a(\eta) = \frac{c}{\eta},
\end{equation}
where $c$ is a positive constant.  In this case, we can obtain an analytic 
solution for HRRT surfaces, which are geodesics in $2+1$ dimensions.

In order to find a spacelike geodesic anchored to two points on the 
leaf, we can use the symmetry of the problem to rotate our axes so that 
the points lie at constant $y = y_0$.  To find a geodesic, we need to 
extremize the distance functional
\begin{equation}
  {\cal D} = \int\!d\eta\, \frac{c}{\eta} \sqrt{\dot{x}^2-1},
\label{eq:L-distance}
\end{equation}
where $\dot{x} = dx/d\eta$, and we have used the fact that the geodesic 
lies on the $y=y_0$ hypersurface.  This functional has no explicit 
dependence on $x$, which means the existence of a quantity that is 
conserved along the geodesic
\begin{equation}
  \frac{\partial{\cal D}}{\partial\dot{x}} 
  = \frac{c\dot{x}}{\eta\sqrt{\dot{x}^2-1}} 
  \equiv p_x.
\label{eq:p_x}
\end{equation}
Using this, we obtain a first-order ordinary differential equation
\begin{equation}
  \frac{d\eta}{dx} = \sqrt{1-\frac{c^2}{p_x^2 \eta^2}},
\label{eq:ODE}
\end{equation}
which can be easily solved to give the analytic expression for the geodesic
\begin{equation}
  \left\{ \begin{array}{ll}
    \eta(x) &= -\sqrt{x^2 + \frac{c^2}{p_x^2}}, \\
    y(x) &= y_0.
  \end{array} \right.
\label{eq:HRRT-dS2+1-1}
\end{equation}

The holographic screen of FRW universes in the de~Sitter limit lies on
\begin{equation}
  \eta = -\sqrt{x^2+y^2} \equiv -r.
\label{eq:dS_2+1_screen}
\end{equation}
Consider a leaf at $\eta = \eta_* = - r_*$ and a subregion on it specified 
by a half opening angle $\psi$ ($0 \leq \psi \leq \pi$).  The end points 
of the HRRT surface are then at
\begin{equation}
  (x, y) = (\mp \eta_* \sin\psi, -\eta_* \cos\psi).
\label{eq:end-points}
\end{equation}
This can be used to determine $p_x$ and $y_0$ in 
Eq.~(\ref{eq:HRRT-dS2+1-1}), giving the final expression 
for the geodesic
\begin{equation}
  \left\{ \begin{array}{ll}
    \eta(x) &= -\sqrt{x^2 + y_0^2}, \\
    y(x) &= y_0,
  \end{array} \right.
\label{eq:HRRT-dS2+1-2}
\end{equation}
where $y_0 = -\eta_* \cos\psi$.  By varying the angle $\psi$, the HRRT 
surfaces sweep a codimension-1 surface in the bulk, which is indeed the 
future boundary of the causal region of the leaf:
\begin{equation}
  \eta = -r,
\qquad
  0 \leq r \leq r_*\, (=\! -\eta_*).
\end{equation}
These surfaces are depicted in $x$-$y$-$\eta$ space in 
Fig.~\ref{fig:dS_2+1}.
\begin{figure}[t]
\begin{center}
  \includegraphics[height=6.5cm]{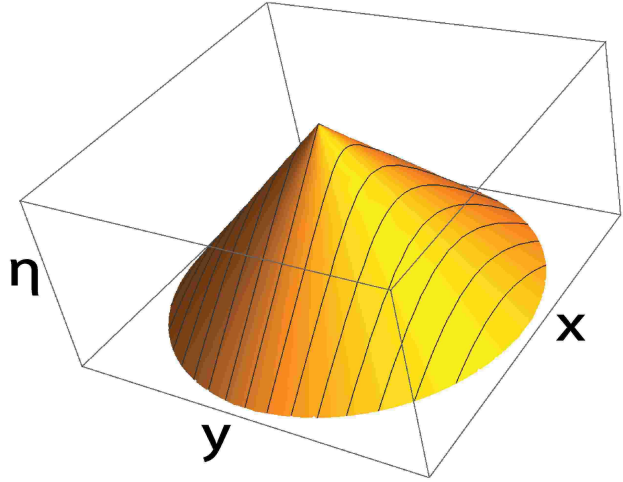}
\end{center}
\caption{HRRT surfaces anchored to subregions on a leaf in 
 $(2+1)$-dimensional de~Sitter space.  They all lie on the future 
 boundary of the causal region associated with the leaf.}
\label{fig:dS_2+1}
\end{figure}
We can clearly see that all the HRRT surfaces are spacelike, except for 
that corresponding to $\psi = \pi/2$ which is null.

\subsubsection*{Small spherical caps in FRW universes in {\boldmath $d+1$} 
 dimensions}

Another example in which simple analytic expressions are obtained 
is the limit of small spherical cap regions, $\psi \ll 1$, on a leaf. 
Consider a flat FRW universe in $d+1$ dimensions
\begin{equation}
  ds^2 = a(\eta)^2 \bigl( -d\eta^2 + dr^2 + r^2 d\Omega_{d-1}^2 \bigr),
\label{eq:FRW_d+1}
\end{equation}
filled with a single fluid component with the equation of state $w$. 
We consider the leaf $\sigma_*$ at $\eta = \eta_*$, which is located at
\begin{equation}
  r = \frac{a}{\dot{a}}.
\label{eq:r_leaf}
\end{equation}
The future boundary $F_*$ of the causal region $D_{\sigma_*}$ is then 
given by
\begin{equation}
  F_*: \eta(r) = \eta_* + \frac{a}{\dot{a}} - r.
\label{eq:F_*}
\end{equation}
Here and below, the scale factor and its derivatives without an argument 
represent those at $\eta = \eta_*$:
\begin{equation}
  a \equiv a(\eta_*),
\qquad
  \dot{a} \equiv \frac{da(\eta)}{d\eta}\biggr|_{\eta = \eta_*},
\qquad
  \ddot{a} \equiv \frac{d^2a(\eta)}{d\eta^2}\biggr|_{\eta = \eta_*}.
\label{eq:a-da-dda}
\end{equation}

We consider a spherical cap region $A$ on the leaf $\sigma_*$, specified 
by a half opening angle $\psi$
\begin{equation}
  0 \leq \theta \leq \psi,
\label{eq:psi}
\end{equation}
where $\theta$ is a polar angle parameterizing $S^{d-1}$ with constant 
$\eta$ and $r$.  Following Ref.~\cite{Nomura:2016ikr}, we go to cylindrical 
coordinates:
\begin{equation}
  \xi = r \sin\theta,
\qquad
  z = r \cos\theta - \frac{a}{\dot{a}} \cos\psi.
\label{eq:cylind}
\end{equation}
In these coordinates, the null cone $F_*$ in Eq.~(\ref{eq:F_*}) is 
given by
\begin{equation}
  F_*: \eta(\xi) = \eta_* + \frac{a}{\dot{a}} 
    - \sqrt{\xi^2 + \Bigl( z + \frac{a}{\dot{a}} \cos\psi \Bigr)^2},
\label{eq:F_*-cylind}
\end{equation}
and the boundary of $A$, $\partial A$, is located at
\begin{equation}
  \eta = \eta_*,
\qquad
  \xi = \frac{a}{\dot{a}} \sin\psi \equiv \xi_*,
\qquad
  z = 0.
\label{eq:boundary-A}
\end{equation}
The HRRT surface $\gamma_A$ anchored to $\partial A$ is on the $z = 0$ 
hypersurface~\cite{Nomura:2016ikr}.  We would like to compare this HRRT 
surface with the intersection of $F_*$ and $z = 0$:
\begin{equation}
  l_A: \eta(\xi) = \eta_* + \frac{a}{\dot{a}} 
    - \sqrt{\xi^2 + \frac{a}{\dot{a}} \cos\psi},
\label{eq:l_A}
\end{equation}
see Fig.~\ref{fig:small-psi}.
\begin{figure}[t]
\begin{center}
  \includegraphics[height=6.5cm]{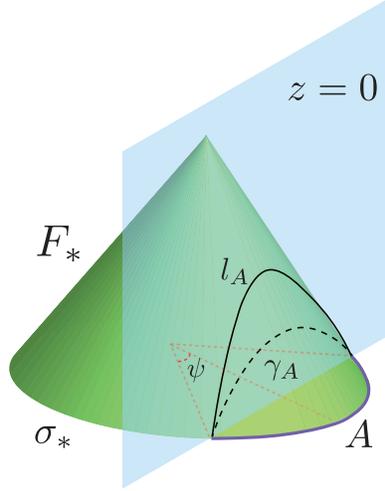}
\end{center}
\caption{The HRRT surface $\gamma_A$ for subregion $A$ of a leaf 
 $\sigma_*$ specified by a half opening angle $\psi$ is on the $z=0$ 
 hypersurface.  It approaches the surface $l_A$, the intersection 
 of the null cone $F_*$ and the $z=0$ hypersurface, in the de~Sitter 
 limit.}
\label{fig:small-psi}
\end{figure}
Using Eq.~(\ref{eq:boundary-A}) and expanding in powers of $\psi \sim 
\xi/(a/\dot{a})$, this can be written as
\begin{equation}
  l_A: \eta(\xi) = \eta_* + \frac{\dot{a}}{2a} (\xi_*^2-\xi^2) 
    + \frac{\dot{a}^3}{8a^3} (\xi_*^2-\xi^2)^2 + \cdots.
\label{eq:l_A-exp}
\end{equation}

For $\psi \ll 1$, the HRRT surface can be expressed in a power series form
\begin{equation}
  \gamma_A: \eta(\xi) 
  = \eta_* + \eta^{(2)}(\xi) + \eta^{(4)}(\xi) + \cdots,
\label{eq:HRRT-exp}
\end{equation}
where
\begin{align}
  \eta^{(2)}(\xi) =& \frac{\dot{a}}{2a} (\xi_*^2 - \xi^2),
\label{eq:exp-2}\\
  \eta^{(4)}(\xi) =& -\frac{\dot{a}}{8a^3(d+1)} (\xi_*^2 - \xi^2) 
\nonumber\\
  & {}\quad 
    \times \Bigl[ \dot{a}^2 \bigl\{ (d+5) \xi_*^2 - (d-3) \xi^2 \bigr\} 
    - a \ddot{a} \bigl\{ (d+3) \xi_*^2 - (d-1) \xi^2 \bigr\} \Bigr].
\label{eq:exp-4}
\end{align}
In the universe dominated by a single fluid component, the scale factor 
behaves as
\begin{equation}
  a(\eta) \propto \eta^{\frac{2}{d-2+dw}}.
\label{eq:a-eta_w}
\end{equation}
Plugging this into Eq.~(\ref{eq:exp-4}), we obtain
\begin{equation}
  \eta^{(4)}(\xi) = \frac{\dot{a}^3}{16(d+1) a^3} (\xi_*^2 - \xi^2) 
    \Bigl[ \bigl\{ 2-(1+3w)d-(1+w)d^2 \bigr\} \xi_*^2 
    - \bigl\{ 2+(3+w)d-(1+w)d^2 \bigr\} \xi^2 \Bigr].
\label{eq:exp-4_w}
\end{equation}
We find that for $w = -1$, the surface given by 
Eqs.~(\ref{eq:HRRT-exp},~\ref{eq:exp-2},~\ref{eq:exp-4_w}) agree with 
$l_A$ in Eq.~(\ref{eq:l_A-exp}).  Namely, the HRRT surface $\gamma_A$ 
is on the null cone $F_*$.

One can see how $\gamma_A$ approaches $F_*$ as $w \rightarrow -1$ by 
subtracting Eq.~(\ref{eq:HRRT-exp}) from Eq.~(\ref{eq:l_A-exp}):
\begin{align}
  \eta^{l_A}(\xi) - \eta^{\gamma_A}(\xi) 
  &= \frac{\dot{a}^3}{16 a^3} \frac{d}{d+1} (1+w) (\xi_*^2 - \xi^2) 
    \bigl\{ (d+3) \xi_*^2 - (d-1) \xi^2 \bigr\} 
\nonumber\\
  &\geq 0.
\label{eq:Delta-eta}
\end{align}
The inequality is saturated only for $w = -1$ (except at the end points 
at $\xi = \xi_*$).

\end{document}